\documentclass[11pt,leqno]{article}
\usepackage{amsthm,amsmath, mathtools}
\usepackage[round]{natbib}
\usepackage{amsfonts}
\usepackage{ifthen}
\usepackage{enumitem}
\usepackage{amssymb,dsfont,url,color,booktabs}
 \usepackage{tikz}
 \usepackage{setspace}

\usepackage{epstopdf}
\usepackage{chngcntr}
\usepackage{bbm}
\usepackage{array}
\usepackage[export]{adjustbox}[2011/08/13] 
\usepackage{geometry}
\usepackage{rotating}
\usepackage{todonotes}
\usepackage{xspace}
\usepackage{algorithm}
\usepackage[noend]{algpseudocode}
 \usepackage[english]{babel}
\usepackage{bbm}
\usepackage{float}
\usepackage{caption}
\usepackage{multirow}
\usepackage{booktabs}
\usepackage[font={footnotesize}]{caption}
\usepackage{placeins}
\usepackage{subfig} 		

\theoremstyle{plain}
\newtheorem{theorem}{Theorem}[section]

\newtheorem{proposition}[theorem]{Proposition}

\newtheorem{remark}[theorem]{Remark}

\newcommand{\notiz}[1]{\relax}

\newcommand{\zitep}[1]{\relax}

\newcommand{\1}{\mathds 1}            

\newcommand{\Price}[1][]{
		\ifthenelse{\equal{#1}{}}{\mathit{Price}}{\Price{}^{#1}}
	} 

\newlength{\wordlength}

\newcommand{\ul}{\underline}
\newcommand{\ol}{\overline}
\newcommand{\RR}{\mathbb{R}}
\newcommand{\EE}{\mathbb{E}}

\newcommand{\dy}{\text{d}y}

\renewcommand{\cite}{\citet}

\numberwithin{equation}{section}
\numberwithin{figure}{section}
\numberwithin{table}{section}

\begin{document}
\title{Speed-up credit exposure calculations for pricing and risk management}

\bigskip
\author{\textbf{Kathrin Glau$\vphantom{l}^{1}$,} \textbf{Ricardo Pachon$\vphantom{l}^{3}$,} \textbf{Christian P{\"o}tz$\vphantom{l}^{1}$
}
\\\\$\vphantom{l}^{\text{1}}$Queen Mary University of London, UK\\
$\vphantom{l}^{\text{3}}$Credit Suisse, UK
}

\maketitle
\begin{abstract}
We introduce a new method to calculate the credit exposure of European and path-dependent options. The proposed method is able to calculate accurate expected exposure and potential future exposure profiles under the risk-neutral and the real-world measure. Key advantage of is that it delivers an accuracy comparable to a full re-evaluation and at the same time it is faster than a regression-based method. Core of the approach is solving a dynamic programming problem by function approximation. This yields a closed form approximation along the paths together with the option's delta and gamma. The simple structure allows for highly efficient evaluation of the exposures, even for a large number of simulated paths. The approach is flexible in the model choice, payoff profiles and asset classes. We validate the accuracy of the method numerically for three different equity products and a Bermudan interest rate swaption. Benchmarking against the popular least-squares Monte Carlo approach shows that our method is able to deliver a higher accuracy in a faster runtime.
\end{abstract}

	

\section{Introduction}
The credit exposure resulting from two counterparties facing each other on a derivatives deal is the main input in a growing list of calculations, all crucial since the financial crisis of 2007--2008. Credit exposures are used to estimate, for example, counterparty credit risk (and consequently the regulatory capital of financial firms), initial margins of collateralized trades, Credit Valuation Adjustments (CVA), Debit Valuation Adjustments (DVA) and, more recently, Funding Valuation Adjustments (FVA).

The exposure of a trade at time $t$ is defined as
\[
E_t(X_t) =\max\{V_t(X_t),0\},
\]
where $X_t$ is the risk factor that drives the price $V_t$ at time $t$ of a portfolio of derivatives. In essence, the credit exposure calculation projects forward in time the distributions of relevant underlying assets, which follow appropriate stochastic models, and obtains the associated distributions of the values of the derivatives in scope, up to their longest maturity. The specifics of this calculation vary with each application. For example, for CVA and DVA the calculation is performed at netting set  while for FVA is done at portfolio level. For CVA, negative exposures are floored to zero before taking a discounted average under the risk-neutral pricing measure $\mathbb{Q}$. In contrast, in order to quantify credit risk, one needs to assess the distribution of the exposure $E_{t}(X_{t})$ under the real-world measure $\mathbb{P}$. For instance, the upper quantiles at the level of $95$, $97.5$ or $99\%$ are standard quantities in risk management.

The mentioned distributions are usually obtained through Monte Carlo simulation: On some chosen time points, the derivatives are re-evaluated on various scenarios, randomly drawn from the distribution of the underlying asset, and from the resulting distribution the required metric is extracted. See \cite{Gregory2010} and \cite{Green2015} for an overview of credit exposure and its calculation. The crux of the calculation is the repeatedly call of the pricers which can be computationally expensive. When their is no closed form solution for the price of the derivative, e.g. for path-dependent options, a straightforward approach would lead to nested Monte Carlo simulations. Moreover, a often a high number of scenario simulations is required to obtain stable results, precisely for tail distributions. In credit risk management an additional challenge arises from the change of measure, i.e. scenarios are generated under the real-world measure, nonetheless pricing is done under the risk-neutral measure. Hence, additionally to simulating the paths of the underlying under $\mathbb{Q}$ the scenario paths need to be simulated under $\mathbb{P}$. A naive simplification would be to assemble the risk quantities also under the pricing measure $\mathbb{Q}$. As reported in \cite{Stein2016}, \textit{"since the banks are already heavily invested in CVA calculations, it is becoming popular to take this shortcut"}. The analysis of \cite{Stein2016} clearly shows the perils of this approach and emphasis the importance of calculating credit risk quantities under the real-world measure.

In the literature regression based methods are studied in order to avoid nested Monte Carlo simulation, see for instance \cite{Schoeftner2008}, who calculate the exposure and CVA for derivatives without analytic solution (e.g. Bermudan options) based on a modification of the least-squares Monte Carlo approach of \cite{LongstaffSchwartz2001}. For the exposure calculation under the real-world measure in a Black-Scholes type model a change of measure using the Radon–Nikodym density is employed. Furthermore, \cite{KarlssonJainOosterlee2016} and \cite{FengJainKarlssonKandhaiOosterlee2016} apply the stochastic grid bundling method (SGBM) of \cite{JainOosterlee2015} to credit exposure calculation and compare it to a least-squares Monte Carlo algorithm. Their comparison reveals severe deficiencies of the L-S approach. Namely, the L-S price introduces numerical noise that leads to inaccurate exposure, especially in its tail distribution. While the bundling technique in the SGBM is able to reduce the Monte Carlo noise and produce more accurate results, it comes at a significant higher cost. A different method is investigated in \cite{ShenWeideAnderluh2013}, who calculate the exposure for Bermudan options on one asset, based on the COS method for early-exercise options of \cite{FangOosterlee2009}. The method produces accurate results under $\mathbb{Q}$ and $\mathbb{P}$ without any change of measure. However, due to its higher runtimes it is mostly suitable for benchmarking.

In this article we propose a new approach to efficiently compute credit exposures of path-dependent options under both, the risk-neutral and the real-world measure. Our ansatz is based on the dynamic programming formulation of the pricing problem. In each step of the backward time-stepping we approximate the price function by a weighted sum of basis functions as proposed in \cite{GlauMahlstedtPoetz2019}. In the latter article it is shown that highly accurate and fast prices can be obtained by this approach based on a suitable approximation technique such as Chebyshev polynomial interpolation. The approximation of option prices by Chebyshev interpolation has some outstanding qualities: it can be quickly constructed from a few evaluations on a grid of asset values; it is robust and efficient to evaluate; and its accuracy can be tuned even for high orders. More generally, in recent years the promising properties of Chebyshev interpolation have been exploited in several areas, see \cite{Trefethen2013} and the \textit{chebfun} project at \textit{www.chebfun.org}.

Our numerical investigation confirms that the proposed method is able to produce accurate exposure profiles under the risk-neutral and the real-world measure. One major advantage of the approach is that it applies to a large variety of products and models, namely, European and path-dependent options in different asset classes. More specifically, in our numerical experiments we validate the method for three different equity products (European, barrier and Bermudan option) and a Bermuda interest rate swaption. As models we cosider the Black-Scholes and the Merton jump-diffusion stock price models and the Hull-White short rate model. We benchmark our method against a least-squares Monte Carlo approach. The numerical comparison reveals that the proposed method is able to deliver a higher accuracy in an even faster runtime. Comparison with a full re-evaluation shows that the error for both, the expected exposure and the potential future exposure is negligible in relation to the scenario simulation error. 

To summarize, the proposed method combines the accuracy of a full re-evaluation with a speed even faster than regression based methods. Therefore, replacing a least-squares Monte Carlo approach by the proposed method enable a considerable more precise quantification of counterparty credit risk. On the level of a whole trading book this will lead to reliable counterparty risk estimates in a reasonable computing time. For an individual bank accurate assessment of counterparty risk results in lower capital requirements. From the regulator's perspective this reduces systematic risk in the banking sector.

The structure of this paper is as follows. In Section 2, we present the definitions of credit exposure for pricing and risk management. In Section 3 we introduce the new approach and we provide algorithms for the exposure calculation under the risk-neutral and the real-world measure and discuss implementational aspects. Section 4 is devoted to the numerical experiments and Section 5 provides a conclusion and outlook.

\section{Credit exposure for pricing and risk management}\label{sec:credit_exposure}
For risk and capital calculation purposes, the expected exposure (EE) is defined as
\begin{equation}\label{eq:EE_risk}
EE_0^{{risk}}(t) = \mathbb{E}^\mathbb{P}(\max(V_t,0)|\mathcal{F}_0),
\end{equation}
where $\mathbb{P}$ refers to the real-world measure, $\mathcal{F}_0$ is the filtration at $t=0$, and $V_t$ is the value of the derivative at time $t$. The potential future exposure of derivative is defined as
\begin{align}\label{eq:PFE_risk}
PFE_0^{risk} (t) =\inf \{y: \mathbb{P}(\max(V_t,0)\leq y)\geq \alpha\}.
\end{align}
for a level $\alpha\in(0,1)$. The class of path-dependent derivatives that we consider in this paper are characterised by a set of exercise dates $t_0,\ldots,t_n = T$, and the value function $V_{t_u}(x)$ of the form
\begin{align}\label{eq:dpp_f}
\begin{split}
	V_T(x) & = g(x),\\
V_{t_u}(x) & = f\Bigl(g(t_u,x),\mathbb{E}^{\mathbb{Q}}[D(t_u,t_{u+1})V_{t_{u+1}}(X_{t_{u+1}})|X_{t_u} = x]\Bigr),
\end{split}
\end{align}
where $f:\mathbb{R}\times\mathbb{R}\rightarrow\mathbb{R}$ is a Lipschitz continuous function, and $g:[0,T]\times \mathbb{R}\rightarrow\mathbb{R}$, with $g(T,x) = g(x)$. Here $X_t$ is the underlying risk factor, and $D(t_u,t_{u+1}) = B_{t_u}/B_{t_{u+1}}$ is the discount factor between $t_u$ and $t_{u+1}$, where $B(t)$ is the bank account
\begin{align}\label{eq:bank_acc}
B(t)=B(0)\exp\Big(\int_{0}^{t}r(s)\text{d}s\Big) \quad \text{with}\quad B(0)=1.
\end{align}
with $r$ the money markets continuously compounded interest rate, and $t<T$. Among the derivatives that can be expressed in the above form, we highlight three that we will use to test our methodology: Classical European options, early-exercise options (Bermudan options) and barrier options.\\

\textbf{Bermudan options:}\\
In this case the value function is given as
\begin{align*}
V_{t_{u}}(x)=\max\left\{g(x),\EE^{\mathbb{Q}}\left[D(t_{u},t_{u+1})V_{t_{u+1}}(X_{t_{u+1}})\vert X_{t_{u}}=x\right]\right\}.
\end{align*}

\textbf{European options:}\\
European options correspond to Bermudan options with no early exercise. In this case the value function becomes
\begin{align*}
V_{t_{u}}(x)=\EE^{\mathbb{Q}}\left[D(t_{u},t_{u+1})V_{t_{u+1}}(X_{t_{u+1}})\vert X_{t_{u}}=x\right].
\end{align*}

\textbf{Barrier options:}\\
Discretely monitored up-and-out barrier option with barrier $B$ can be written in the same form with value function
\begin{align*}
V_{t_{u}}(x)=\EE^{\mathbb{Q}}\left[D(t_{u},t_{u+1})V_{t_{u+1}}(X_{t_{u+1}})\vert X_{t_{u}}=x\right]\1_{x\leq B}.
\end{align*}
Similarly, we can use the framework for down-and-out barrier options.\\

The expected exposure also appears when pricing the basis between the counterparty risk-free value of a trade and its valuation when accounting for counterparty risk. This difference arises from the risk that a trade is in favour of one counterparty but the other one defaults before the trade matures. This Credit Valuation Adjustment (CVA) is equivalent to the price of a contingent CDS, whose value follows from the fundamental arbitrage theorem: 
\[
\frac{\text{CVA}_0}{B(0)} = \mathbb{E}^{\mathbb{Q}}\Biggl[ \int_{s=0}^{s=T} \frac{\max(V_s,0)\cdot d\mathbf{1}_{(\tau\leq s)}}{B(s)}\Biggr]=\int_{s=0}^{s=T} \mathbb{E}^{\mathbb{Q}} \Biggl[\frac{\max(V_s,0)\cdot d\mathbf{1}_{(\tau\leq s)}}{B(s)}\Biggr],
\]
where $\mathbb{Q}$ is the associated risk-neutral measure and $\mathbf{1}_{(\tau\leq s)}$ is the default indicator for the counterparty which equals 1 if $s$ is less than the default time $\tau$ and 0 otherwise. The integral over time can be discretized over time buckets, and in the special case that the value of the derivative and the default event are independent, the expectation can be expressed as the product of two terms, one accounting exclusively for the default probability and the other one for the positive exposure of the trade. This exposure is calculated as
\begin{equation}\label{eq:EE_price}
EE_0^{price}(t) = \mathbb{E}^\mathbb{Q} \Bigl(\frac{\max(V_t,0)}{B(t)}\Bigl|\mathcal{F}_0\Bigr) = D(0,t) \mathbb{E}^\mathbb{Q} (\max(V_t,0)|\mathcal{F}_0) ,
\end{equation}
assuming that $B(0)=1$. Moreover, we define the $\mathbb{Q}$-counterpart of $PFE^{risk}$ as  
\begin{align}\label{eq:PFE_risk}
PFE_0^{price} (t) =\inf \{y: \mathbb{Q}(D(0,t)\max(V_t,0)\leq y)\geq \alpha\}.
\end{align}
The differences between the risk and the pricing exposures, i.e., expressions (\ref{eq:EE_risk}) and (\ref{eq:EE_price}), is that the former uses the real-world measure for diffusing the risk factors, while the later uses the risk-neutral measure (the pricing of $V_t$ in both cases, of course, uses $\mathbb{Q}$).  Additionally, for pricing exposures we also incorporate a discount factor at time point $t$. As we will see in Section \ref{seq:exposure_pricing} and \ref{seq:exposure_risk}, the structure of our methodology does not change much when calculating either one of them. 

\section{A unified approach for exposure calculation}
In this section we presented a unified approach for the calculation of credit exposure for different types of path-dependent options. The core idea of our approach is to write the option price as a solution of a Dynamic Programming problem and to approximate the solution with a suitable set of basis functions. The proposed ansatz is based on the dynamic Chebyshev algorithm of \cite{GlauMahlstedtPoetz2019}. This method was presented as a pricing method and can be very easily extended to calculate expected exposures of options.

\subsection{Calculating credit exposures using dynamic programming}
For many (portfolios of) derivatives the expected exposure as defined in \eqref{eq:EE} cannot be calculated analytically and simulation approaches come into play. The risk factors $X_{t}^{i}$, $i=1,\ldots,M$ are simulated and the expected exposure is approximated by
\begin{align*}
EE_{t}(x)=\EE^{\mathbb{P}}[\max\{V_{t}(X_{t}),0\}]\approx\frac{1}{M}\sum_{i=1}^{M}\max\{V_{t}(X_{t}^{i}),0\}.
\end{align*}
Hence, the values $V_{t}(X_{t}^{i})$ of the derivative have to be calculated for a large number $M$ of simulated risk factors. Typically, there is no analytic solution available and the evaluation becomes computationally demanding. This is especially the case when the value function $V_{t}$ at time point $t$ depends on the conditional expectation of the value function at $t+1$. 

In order to address this issue we propose to approximate the function $x\mapsto V_{t}(x)$ with a weighted sum of basis functions, i.e.
\begin{align*}
V_{t}(x)\approx\widehat{V}_{t}(x)=\sum_{j=0}^{N}c_{j}p_{j}(x)
\end{align*}
with weights/coefficients $c_{j}$. Then we replace the value function with its approximation in the exposure calculation
\begin{align}\label{eq:EE_calc_Cheb_approx}
EE_{t}(x)=\EE^{\mathbb{P}}[\max\{V_{t}(X_{t}),0\}]\approx\frac{1}{M}\sum_{i=1}^{M}\max\{\widehat{V}_{t}(X_{t}^{i}),0\}.
\end{align}
Even for a large number of simulated risk factors the sum of basis functions can be evaluated efficiently.\\ 

In order to introduce the algorithm, we start with the pricing of a Bermudan option. The value of a Bermudan option with payoff $g$ and exercise dates $t_{0},\ldots,t_{n}=T$ is given by the optimal stopping problem
\begin{align*}
V_{t_{0}}(x)=\sup_{t_{0}\leq t_{u}\leq T}B(t_{0})\EE^{\mathbb{Q}}\Big[\frac{g(X_{t_{u}})}{B(t_{u})}\Big\vert X_{t_{0}}=x\Big]
\end{align*}
where $B(t)$ is the bank account given by \eqref{eq:bank_acc}. The principle of Dynamic Programming yields the backward induction
\begin{align*}
&V_{T}(x)=g(x)\\
&V_{t_{u}}(x)=\max\left\{g(x),\EE^{\mathbb{Q}}\left[D(t_{u},t_{u+1})V_{t_{u+1}}(X_{t_{u+1}})\vert X_{t_{u}}=x\right]\right\}
\end{align*}
for the discount factor $D(t_{u},t_{u+1})=B_{t_{u}}/B_{t_{u+1}}$. More generally, we obtain the dynamic programming problem \eqref{eq:dpp_f} for the value function $V_{t_{u}}(x)$ 
\begin{align*}
V_{t_{u}}(x)=f\left(g(t_{u},x),\EE^{\mathbb{Q}}\left[D(t_{u},t_{u+1})V_{t_{u+1}}(X_{t_{u+1}})\vert X_{t_{u}}=x\right]\right)
\end{align*}
for a Lipschitz continuous function $f:\RR\times\RR\rightarrow\RR$ and a function $g:[0,T]\times\RR\rightarrow\RR$ with $g(T,x)=g(x)$. This formulation includes also the pricing of European and barrier options, as stated in Section \ref{sec:credit_exposure}.\\

We will solve the backward induction on the finite domain $\mathcal{X}=[\ul{x},\ol{x}]$. Assume we have at $t_{u+1}$ an approximation $\widehat{V}_{t_{u+1}}$ with $V_{t_{u+1}}(x)\approx\widehat{V}_{t_{u+1}}(x)=\sum_{j}c_{j}(t_{u+1})p_{j}(x)$. We solve the problem on a set of nodal points $x_{k}$, $k=0,\ldots,N$ and use the function values at these points to calculate new coefficients $c_{j}$. In this case the backward induction becomes
\begin{align*}
V_{t_{u}}(x_{k})&=f\left(g(t_{u},x_{k}),\EE^{\mathbb{Q}}\left[D(t_{u},t_{u+1})V_{t_{u+1}}(X_{t_{u+1}})\vert X_{t_{u}}=x_{k}\right]\right)\\
&\approx f\Big(g(t_{u},x_{k}),\EE^{\mathbb{Q}}\Big[D(t_{u},t_{u+1})\sum_{j=0}^{N}c_{j}(t_{u+1})p_{j}(X_{t_{u+1}})\vert X_{t_{u}}=x_{k}\Big]\Big)\\
&=f\Big(g(t_{u},x_{k}),\sum_{j=0}^{N}c_{j}(t_{u+1})\EE^{\mathbb{Q}}\left[D(t_{u},t_{u+1})p_{j}(X_{t_{u+1}})\vert X_{t_{u}}=x_{k}\right]\Big),
\end{align*}
where we exploited the linearity of the conditional expectation. Here, we see that the coefficients $c_{j}$ carry the information of the payoff, and the conditional expectations $\EE^{\mathbb{Q}}\left[D(t_{u},t_{u+1})p_{j}(X_{t_{u+1}})\vert X_{t_{u}}=x_{k}\right]$ carry the information of the stochastic process. Since the conditional expectations are independent of the backward induction they can be pre-computed in an offline step before the actual pricing. In the section section, we will propose a suitable set of nodal points and basis function and explain how to obtain the coefficients $c_{j}$ in every time step.\\

The presented procedure is a pricing method for a large class of option pricing problems which can be written in the form of \eqref{eq:dpp_f}. This includes different option types, payoff profiles as well as different asset classes and models.\\

Now, we are in a position to efficiently evaluate the exposure in formula \eqref{eq:EE_calc_Cheb_approx}. Assume we have simulated $M$ paths of the underlying risk factor. Then we price the option along the paths using the closed form approximation
\begin{align*}
V_{t_{u}}(X_{t_{u}}^{i})\approx\sum_{j=0}^{N}c_{j}(t_{u})p_{j}(X_{t_{u}}^{i}) \quad \text{for} \quad i=1,\ldots,M \ \text{ and } \ u=0,\ldots,n.
\end{align*}
These values can now be used to calculate the expected exposure or the potential future exposure for a given level $\alpha$. In the case of a Bermudan option one has to take into account that by exercising the option at $t_{u}$ the exposure becomes zero. Similarly, if the barrier option is knocked out the exposure at all future time steps is zero. These two effects yield a decreasing exposure for both types of options.\\ 

\textbf{Discounting:}\\
If he interest rate is our risk factor, i.e. $r(t)=r(t,X_{t})$, we simplify the expectation of the discounted basis function in the following way. Assume that the time stepping $\Delta t=t_{u+1}-t_{u}$ is small, then we can write
\begin{align*}
\EE^{\mathbb{Q}}&\left[D(t_{u},t_{u+1})p_{j}(X_{t_{u+1}})\vert X_{t_{u}}=x_{k}\right]
\\&=\EE^{\mathbb{Q}}\left[\exp\Big(-\int_{t_{u}}^{t_{u+1}}r(s,X_{s})\text{d}s\Big)p_{j}(X_{t_{u+1}})\vert X_{t_{u}}=x_{k}\right]\\
&\approx\EE^{\mathbb{Q}}\left[\exp\big(-\Delta t r(t_{u},X_{t_{u}})\big)p_{j}(X_{t_{u+1}})\vert X_{t_{u}}=x_{k}\right]\\
&=\exp\big(-\Delta t r(t_{u},x_{k})\big)\EE^{\mathbb{Q}}\left[p_{j}(X_{t_{u+1}})\vert X_{t_{u}}=x_{k}\right],
\end{align*}
where we assume that the discount factor is constant on a small interval. Otherwise, if the discount factor is deterministic we can simply write
\begin{align*}
\EE^{\mathbb{Q}}\left[D(t_{u},t_{u+1})p_{j}(X_{t_{u+1}})\vert X_{t_{u}}=x_{k}\right]
=D(t_{u},t_{u+1})\EE^{\mathbb{Q}}\left[p_{j}(X_{t_{u+1}})\vert X_{t_{u}}=x_{k}\right].
\end{align*}
In both cases, we only need to pre-compute the expectations $\EE^{\mathbb{Q}}\left[p_{j}(X_{t_{u+1}})\vert X_{t_{u}}=x_{k}\right]$.

\subsection{Choice of basis function and grid points}
Crucial for an efficient algorithm is the choice of an appropriate approximation method, i.e. the choice of basis functions $p_{j}$ and nodal points $x_{k}$. The chosen approximation method should be able to satisfy different requirements. The approximation error of the method should converge uniformly for a large class of (value) functions. A smooth value function should yield a fast error decay and good approximation results for a relatively low number of nodal points. The method should provide an efficient way to compute the coefficients $c_{j}$, ideally using an explicit formula. For the exposure calculation, the evaluation of the sum $\sum_{j}c_{j}p_{j}$ needs to be done in a fast and numerically stable way even for large sets of input values. A suitable choice for this task is Chebyshev polynomial interpolation as proposed by \cite{GlauMahlstedtPoetz2019}.

\subsubsection{Chebyshev polynomial interpolation}
The one-dimensional Chebyshev interpolation is a polynomial interpolation of a function $f$ in the interval $[-1,1]$ of degree $N$ in the $N+1$ Chebyshev points $z_{k}=\cos(\pi k/N)$. These points are not equidistantly distributed but cluster at $-1$ and $1$. The interpolant can be written as a sum Chebyshev polynomials $T_{j}(z)=\cos(j\,\text{acos}(z))$ with an explicit formula for the coefficients, i.e. for a function $f:[-1,1]\rightarrow\RR$ we obtain
\begin{align*}
I_{N}(f)(z)=\sum_{j=0}^{N}c_{j}T_{j}(z) \quad \text{with} \quad c_{j}=\frac{2^{\1_{\{0<j<N\}}}}{N}\sum_{k=0}^{N}{}^{''}f(z_{k})T_{j}(z_{k})
\end{align*}
where $\sum{}^{''}$ indicates the summand is multiplied by $1/2$ if $k=0$ or $k=N$. In order to evaluate the interpolation efficiently one can exploit the following alternative definition of the Chebyshev polynomials
\begin{align}\label{Chebpoly_recur}
T_{n+1}(z)=2zT_{n}(z)-T_{n-1}(z), \qquad T_{1}(z)=z \quad \text{and} \quad T_{0}(z)=1.
\end{align}
Based on this recurrence relation Clenshaw's algorithm provides an efficient framework to evaluate the Chebyshev interpolant $I_{N}(f)$
\begin{align*}
&b_{k}(x)=c_{k}+ 2xb_{k+1}(x)-b_{k+2}(x),\quad\text{for}\quad k=n,\ldots,1\\
&I_{N}(f)(x)=c_{0}+xb_{1}(x)-b_{2}(x)
\end{align*}
with starting values $b_{n+1}(x)=b_{n+2}(x)=0$.

In order to interpolate functions on an arbitrary rectangular $\mathcal{X}=[\underline{x},\overline{x}]$, we introduce a transformation $\tau_{\mathcal{X}}:[-1,1]\rightarrow\mathcal{X}$ defined by
\begin{align}
\tau_{\mathcal{X}}(z)=\overline{x}+0.5(\underline{x}-\overline{x})(1-z).\label{Transformation}
\end{align}
The Chebyshev interpolation of a function $f:\mathcal{X}\rightarrow\mathbb{R}$ can be written as
\begin{align}\label{Cheby_Interpolation}
I_{\overline{N}}(f)(x)=\sum_{j=0}^{N}c_{j}p_{j}(x) \quad \text{with}\quad c_j&=\frac{2^{\1_{\{0<j<N\}}}}{N_i}\sum_{k=0}^{N}{}^{''}f(x_k)T_{j}(z_k)
\end{align}
for $x\in\mathcal{X}$ with transformed Chebyshev polynomials $p_j(x)=T_j(\tau^{-1}_{\mathcal{X}}(x))1_{\mathcal{X}}(x)$ and transformed Chebyshev points $x_k=\tau_{\mathcal{X}}(z_k)$.
The one-dimensional interpolation has a tensor based extension to the multivariate case, see e.g. \cite{SauterSchwab2010}.\\

The Chebyshev interpolation provides promising convergence results and explicit error bounds. The interpolation converges for all Lipschitz continuous functions and for analytic functions the interpolation converges exponentially fast. See \cite{Trefethen2013} for the one-dimensional case and for a multivariate version \cite{SauterSchwab2010}. Moreover, the convergence is of polynomial order for differentiable functions and the derivatives converge as well, see \cite{GassGlauMahlstedtMair2018}. 

The Chebyshev interpolation is implemented in the open-source $\textit{MATLAB}$ package $\textit{chebfun}$ available at $\textit{www.chebfun.org}$.

\subsection{The Dynamic Chebyshev algorithm for exposure calculation}
Using Chebyshev interpolation as approximation technique, the time step in the backward induction looks as follow. Assume we have the nodal values $V_{t_{u}}(x_{k})$ at the Chebyshev points $x_{k}$, $k=0,\ldots,N$. Then, the explicit formula for the coefficients of the Chebyshev interpolation yields
\begin{align*}
c_{j}(t_{u})=\frac{2^{1_{0<j<N}}}{N}\sum_{k=0}^{N}{}^{''}V_{t_{u}}(x_{k})T_{j}(z_{k}),
\end{align*}
and we obtain a closed form approximation of the option price
\begin{align*}
V_{t_{u}}(x)\approx\widehat{V}_{t_{u}}(x)=\sum_{j=0}^{N}c_{j}(t_{u})p_{j}(x).
\end{align*}
Note that we presented the framework for an option on one underlying. In case of multiple underlyings we only need to replace the one-dimensional Chebyshev interpolation with its multivariate extension. The more general multivariate version of the algorithm is presented in \cite{GlauMahlstedtPoetz2019}.\\

The resulting pricing algorithm is for all three option types (Bermudan, barrier, European) essentially the same. However, the efficiency of the method is directly related to the smoothness of the value function. As a result the number of nodal points required for a given accuracy varies, compare Section 5.2 and 5.3 in \cite{GlauMahlstedtPoetz2019}. Moreover, the size of the interpolation domain influences the number of nodal points that a required for a target accuracy.\\

The resulting algorithm for exposure calculation under the pricing measure and under the real-world measure with the dynamic Chebyshev method is presented in the following two sections.

\subsubsection{Exposure calculation for pricing}\label{seq:exposure_pricing}
Here, we consider the computation of the exposure under the pricing measure $\mathbb{Q}$. The main application is the computation of the expected exposure $EE^{price}_{0}$ as an ingredient of the CVA calculation.\\

\textbf{Algorithm: Exposure of Bermudan options under $\mathbb{Q}$}\\
This algorithm provides a framework to calculate the expected exposure and the potential future exposure for a Bermudan option. A European option can be seen as a special case and falls also in the scope of this algorithm.\\

\textit{1. Simulation of risk factors}:\\
Simulate $M$ paths of the underlying risk factor $X_{t_{0}}^{i},\ldots,X_{t_{n}}^{i}$, $i=1,\ldots,M$ under the pricing measure $\mathbb{Q}$.\\

\textit{2. Preparation of the pricing algorithm}:\\
Find a suitable interpolation domain $\mathcal{X}=[\ul{x},\ol{x}]$ and calculate the nodal points $x_{k}=\tau_{\mathcal{X}}(\cos(k\pi/N))$, $k=0,\ldots,N$ for this domain. Pre-compute the conditional expectations of the basis function under the pricing measure $\mathbb{Q}$
\[\Gamma_{k,j}=\EE^{\mathbb{Q}}\left[p_{j}(X_{\Delta t})\vert X_{0}=x_{k}\right].\]

\textit{3. Initialization of the pricing algorithm}:\\
Start pricing at maturity $T$ and compute nodal values $\widehat{V}_{T}(x_{k})=g(T, x_{k})$ for all $k=0,\ldots,N$ for the payoff function $g(T,x_{k})$. Calculate Chebyshev coefficients $c_{j}(T)$ using the nodal values $\widehat{V}_{T}(x_{k})$. For all paths compute the exposure $E_{T}^{i}=\max\{g(T, X_{T}^{i}),0\}$.\\

\textit{4. Exposure calculation via backward induction}:\\
Iterative time stepping $t_{u+1}\rightarrow t_{u}$: Assume we have a Chebyshev approximation $V_{t_{u+1}}(x)\approx\widehat{V}_{t_{u+1}}(x)=\sum_{j}c_{j}(t_{u+1})p_{j}(x)$
\begin{itemize}
\item compute nodal values 
\[\widehat{V}_{t_{u}}(x_{k})=
\begin{cases}
\max\{g(x_k),D_{u}(x_{k})\sum_{j=0}^{N}c_{j}(t_{u+1})\Gamma_{k,j}\},\quad&\text{if }t_{u}\text{ is exercise day},\\
D_{u}(x_{k})\sum_{j=0}^{N}c_{j}(t_{u+1})\Gamma_{k,j},\quad&\text{otherwise}
\end{cases}
\]
with discount factor $D_{u}(x_k)=D(t_{u},t_{u+1},x_{k})$
\item calculate new coefficients $c_{j}(t_u)$ using nodal values $\widehat{V}_{t_{u}}(x_{k})$,
\item price the option for all simulation paths $V_{t_{u}}^{i}=\widehat{V}_{t_u}(X_{t_{u}}^{i})=\sum_{j\in J}c_{j}(t_{u})p_{j}(X_{t_{u}}^{i})$,
\item calculate exposure $E_{t_{u}}^{i}=D(0,t_{u})\max\{V_{t_{u}}^{i},0\}$,
\item if the option is exercised (i.e $V_{t_{u}}^{i}=g(X_{t_{u}}^{i})$), update the exposure at all future time steps on this path $E_{t_{j}}^{i}$, $j=u+1,\ldots,n$.
\end{itemize}

\textit{5. Calculation of expected exposure}:\\
Obtain an approximation of the expected future exposures
\[EE^{price}_{0}(t_{u})=\EE^{\mathbb{Q}}\left[D(t_{u})\max\{V_{t_u},0\}\right]\approx\frac{1}{M}\sum_{i=1}^{M}D(t_{u})E_{t_{u}}^{i},\]
and an approximation of the potential future exposures
\[PFE_{0}^{price}(t_{u})=\inf\big\{y:\mathbb{Q}\big(E_{t}(x)\leq y\big)\geq \alpha\big\}\approx\inf\big\{y:\frac{\#\{E^{i}_{t_{u}}\leq y\}}{M}\geq \alpha\big\}.\]
for all $u=0,\ldots,n$.\\[2ex]

\textbf{Modification of the algorithm for barrier options}\\
The presented algorithm Bermudan option can be modified to calculate the exposure of barrier options. In this case the interpolation domain is chosen depending on the barrier. There is no early exercise, however, we need to take care of the knock-out feature. For an up-and-out option with barrier $B$ and $b=\log(B)$, the following modifications are added to the algorithm. First, the interpolation domain is set as $\mathcal{X}=[\ul{x},b]$. Second, the iterative time stepping from $t_{u+1}\rightarrow t_{u}$ is modified in the following way. Assume we have a Chebyshev approximation $V_{t_{u+1}}(x)\approx\widehat{V}_{t_{u+1}}(x)=\sum_{j}c_{j}(t_{u+1})p_{j}(x)$,
\begin{itemize}
\item compute nodal values $\widehat{V}_{t_{u}}(x_{k})=D_{u}(x_{k})\sum_{j=0}^{N}c_{j}\Gamma_{k,j}$ and new coefficients $c_{j}(t_u)$,
\item price the option for all simulation paths $V_{t_{u}}^{i}=\widehat{V}_{t_u}(X_{t_{u}}^{i})=\sum_{j\in J}c_{j}(t_{u})p_{j}(X_{t_{u}}^{i})$ if $X_{t_{u}}^{i}\leq b$ and $V_{t_{u}}^{i}=0$ otherwise,
\item calculate the exposure $E_{t_{u}}^{i}=D(0,t_{u})\max\{V_{t_{u}}^{i},0\}$,
\item if the option is knocked-out, i.e if $X_{t_{u}}^{i}> b$ update the exposure at all future time steps on this path $E_{t_{j}}^{i}$, $j=u+1,\ldots,n$.
\end{itemize}

\subsubsection{Exposure calculation for risk management}\label{seq:exposure_risk}
In this section we present an algorithm for the exposure calculation under the real-world measure $\mathbb{P}$.\\

\textbf{Algorithm: Exposure of Bermudan options under $\mathbb{P}$}\\
This algorithm provides a framework to calculate the expected exposure and the potential future exposure for a Bermudan option.
\textit{1. Simulation of risk factors}:\\
Simulate $M$ paths of the underlying risk factor $X_{t_{0}}^{i},\ldots,X_{t_{n}}^{i}$, $i=1,\ldots,M$ under the real-world measure $\mathbb{P}$.\\

\textit{2. Preparation of the pricing algorithm}:\\
Find a suitable interpolation domain $\mathcal{X}=[\ul{x},\ol{x}]$ and calculate the nodal points $x_{k}=\tau_{\mathcal{X}}(\cos(k\pi/N))$, $k=0,\ldots,N$ for this domain. Pre-compute the conditional expectations of the basis function under the pricing measure $\mathbb{Q}$
\[\Gamma_{k,j}=\EE^{\mathbb{Q}}\left[p_{j}(X_{\Delta t})\vert X_{0}=x_{k}\right].\]

\textit{3. Initialization of the pricing algorithm}:\\
Start pricing at maturity $T$ and compute nodal values $\widehat{V}_{T}(x_{k})=g(T, x_{k})$ for all $k=0,\ldots,N$ for the payoff function $g(T,x_{k})$. Calculate Chebyshev coefficients $c_{j}(T)$ using the nodal values $\widehat{V}_{T}(x_{k})$. For all paths compute the exposure $E_{T}^{i}=\max\{g(T, X_{T}^{i}),0\}$.\\

\textit{4. Exposure calculation via backward induction}:\\
Iterative time stepping $t_{u+1}\rightarrow t_{u}$: Assume we have a Chebyshev approximation $V_{t_{u+1}}(x)\approx\widehat{V}_{t_{u+1}}(x)=\sum_{j}c_{j}(t_{u+1})p_{j}(x)$
\begin{itemize}
\item compute nodal values 
\[\widehat{V}_{t_{u}}(x_{k})=
\begin{cases}
\max\{g(x_k),D_{u}(x_{k})\sum_{j=0}^{N}c_{j}(t_{u+1})\Gamma_{k,j}\},\quad&\text{if }t_{u}\text{ is exercise day},\\
D_{u}(x_{k})\sum_{j=0}^{N}c_{j}(t_{u+1})\Gamma_{k,j},\quad&\text{otherwise}
\end{cases}
\]
with discount factor $D_{u}(x_k)=D(t_{u},t_{u+1},x_{k})$
\item calculate new coefficients $c_{j}(t_u)$ using nodal values $\widehat{V}_{t_{u}}(x_{k})$,
\item price the option for all simulation paths $V_{t_{u}}^{i}=\widehat{V}_{t_u}(X_{t_{u}}^{i})=\sum_{j\in J}c_{j}(t_{u})p_{j}(X_{t_{u}}^{i})$,
\item calculate exposure $E_{t_{u}}^{i}=\max\{V_{t_{u}}^{i},0\}$,
\item if the option is exercised (i.e $V_{t_{u}}^{i}=g(X_{t_{u}}^{i})$), update the exposure at all future time steps on this path $E_{t_{j}}^{i}$, $j=u+1,\ldots,n$.
\end{itemize}

\textit{5. Calculation of expected exposure}:\\
Obtain an approximation of the expected future exposures
\[EE^{risk}_{0}(t_{u})=\EE^{\mathbb{P}}\left[D(t_{u})\max\{V_{t_u},0\}\right]\approx\frac{1}{M}\sum_{i=1}^{M}D(t_{u})E_{t_{u}}^{i},\]
and an approximation of the potential future exposures
\[PFE_{0}^{risk}(t_{u})=\inf\big\{y:\mathbb{P}\big(E_{t}(x)\leq y\big)\geq \alpha\big\}\approx\inf\big\{y:\frac{\#\{E^{i}_{t_{u}}\leq y\}}{M}\geq \alpha\big\}.\]
for all $u=0,\ldots,n$.\\[2ex]
Similarly to the exposure calculation for pricing we can modify the algorithm or barrier options.\\

A comparison with the algorithms in the previous section shows that the exposure calculation under $\mathbb{Q}$ and $\mathbb{P}$ has the same structure. The difference is that the paths of the risk factor(s) are simulated under a different measure and the calculation of $EE_{0}^{price}(t)$ requires the discount factor at time point $t$. Moreover, if we are interested in the PFE we need a higher number of simulation paths since the PFE is a tail measure.

\subsection{Conceptional benefits of the method}
The presented algorithms provide efficient solutions for the exposure calculation. Moreover, the structure of the new approach comes with conceptual benefits, which can be exploited in practice.\\

\textbf{Error analysis}\\
Let $\varepsilon_{t}:=\Vert V_{t}-\widehat{V}_{t}\Vert_{\infty}=\max_{x\in\mathcal{X}}\vert V_{t}(x)-\widehat{V}_{t}(x)\vert$ be the error of the dynamic Chebyshev method and assume that the truncation error for this domain is negligible. From \cite{GlauMahlstedtPoetz2019} we obtain the following result for the convergence of the dynamic Chebyshev method. If the value function $x\mapsto V_{t}(x)$ is analytic, the log-error decays nearly linearly in the number of nodal points $N$, i.e.
\begin{align}\label{eq:log_error_analytic}
\log(\varepsilon_{t})\leq -c_{1}N + c_{2}\log(\log(N)) + c_{3}
\end{align}
for constants $c_{1},c_{2},c_{3}>0$. If the value function is $p$-times continuously differentiable the 
log-error decays nearly linearly in the logarithm of the number of nodal points $N$, i.e.
\begin{align}\label{eq:log_error_diff}
\log(\varepsilon_{t})\leq -p\log(N) + c_{2}\log(\log(N)) + c_{3}.
\end{align}
The analyticity of the value function holds for European and barrier options, whereas, the value function of a Bermudan option is only continuously differentiable.

In practice, a convenient approach to assess convergence a posteriori is to investigate the decline of the absolute values of the estimated coefficients $\vert c_{j}\vert$, see for instance the implementation of the \textit{chebfun} package on \textit{www.chebfun.org}.\\

The error analysis for the pricing can be directly applied to the exposure calculation. Assume the error $\Vert V_{t}-\widehat{V}_{t}\Vert_{\infty}<\varepsilon$ for some $\varepsilon>0$. For the expected exposure of an option $EE_{t}=\EE[V_{t}(X_{t})\vert X_{t}=x_{0}]$ holds
\begin{align*}
&\Big\vert \EE[V_{t}(X_{t})\vert X_{t}=x_{0}] - \frac{1}{M}\sum_{i=0}^{M}\widehat{V}_{t}(X_{t}^{i})\Big\vert\\
&\quad\leq \Big\vert \EE[V_{t}(X_{t})\vert X_{t}=x_{0}] - \frac{1}{M}\sum_{i=0}^{M}V_{t}(X_{t}^{i})\Big\vert + \Big\vert\frac{1}{M}\sum_{i=0}^{M}\big(V_{t}(X_{t}^{i}) -\widehat{V}_{t}(X_{t}^{i})\big)\Big\vert\\
&\quad\leq \Big\vert \EE[V_{t}(X_{t})\vert X_{t}=x_{0}] - \frac{1}{M}\sum_{i=0}^{M}V_{t}(X_{t}^{i})\Big\vert +\frac{1}{M}\sum_{i=0}^{M} \big\vert V_{t}(X_{t}^{i}) -\widehat{V}_{t}(X_{t}^{i})\big\vert.
\end{align*}
The first term is the Monte Carlo error of the scenario simulation which decays with $M^{-1/2}$. This error is the same for the full re-evaluation approach and our dynamic Chebyshev approach. The second term is the actual pricing error between a full re-evaluation and our method and is bounded by $\varepsilon$. This derivation holds under both measures $\mathbb{Q}$ and $\mathbb{P}$ and hence for $EE^{price}_{t}$ and $EE^{risk}_{t}$.

In the same way the error for the potential future exposure splits into the scenario simulation error and the pricing error. Here, the PFE using a full re-evaluation and the dynamic Chebyshev method are given by 
\begin{align*}
y^{\star}:=\inf\big\{y:\frac{\#\{V_{t}(X_{t}^{i})\leq y\}}{M}\geq \alpha\big\}, \quad\text{and}\quad \widehat{y}^{\star}:=\inf\big\{y:\frac{\#\{\widehat{V}_{t}(X_{t}^{i})\leq y\}}{M}\geq \alpha\big\}.
\end{align*}
We assume $M\alpha$ is an integer and obtain the equality
\begin{align*}
\#\{V_{t}(X_{t}^{i})\leq y^{\star}\}=\#\{\widehat{V}_{t}(X_{t}^{i})\leq \widehat{y}^{\star}\}=\#\{V_{t}(X_{t}^{i})\leq \widehat{y}^{\star}+\Delta V^{i}\}\leq \#\{V_{t}(X_{t}^{i})\leq \widehat{y}^{\star}+\varepsilon\}
\end{align*}
with $\Delta V^{i} = V_{t}(X_{t}^{i})-\widehat{V}_{t}(X_{t}^{i})$. It follows that $y^{\star}\leq\widehat{y}^{\star}+\varepsilon$ and similarly, exchanging the roles of $V_{t}$ and $\widehat{V}_{t}$ yields $\widehat{y}^{\star}\leq y^{\star}+\varepsilon$. Hence we obtain the difference of the estimated potential future exposure $\vert y^{\star}-\widehat{y}\vert<\varepsilon$.
\begin{remark}
The difference between the expected exposure computed via Chebyshev approximation and via full re-evaluation is bounded by the pricing error $\Vert V_{t}-\widehat{V}_{t}\Vert_{\infty}$. The same holds for the potential future exposure. The pricing error $\Vert V_{t}-\widehat{V}_{t}\Vert_{\infty}$ decays nearly exponentially in the number of points for European and barrier options and nearly algebraically for Bermudan options.\\
\end{remark}

\textbf{Closed form expression for the conditional expectations}\\
The conditional expectations of the Chebyshev polynomials $\EE^{\mathbb{Q}}\left[p_{j}(X_{t_{u+1}})\vert X_{t_{u}}=x_{k}\right]$ depend only on the underlying process and can be pre-computed prior to the time-stepping. Here two different cases have to be distinguished.\\
If the underlying process $X_{t_{u+1}}\vert X_{t_{u}}=x$ is normally distributed the conditional expectations of the Chebyshev polynomials can be calculated analytically. Examples are the Black-Scholes model (with log-stock price $X_t$), the Vasicek model or the one factor Hull-White model (both with interest rate $X_t$). More generaly, assume for instance the underlying process is modelled via an SDE of the form
\begin{align*}
\text{d}X_{t}=\alpha(t,X_{t})\text{d}t + \beta(t,X_{t})\text{d}W_{t}
\end{align*}
for a standard Brownian motion $W_{t}$ with Euler–Maruyama approximation
\begin{align*}
X_{t_{u+1}}\approx x + \alpha(t_u,x)(t_{u+1}-t_{u}) + \beta(t_{u},x)\sqrt{t_{u+1}-t_{u}}Z =:\widehat{X}^{x}_{t_{u+1}} \qquad Z\sim\mathcal{N}(0,1)
\end{align*}
and the right hand side is thus normally distributed. The following proposition provides an analytic formula for the conditional moments $\EE^{\mathbb{Q}}[p_{j}(\widehat{X}^{x_{k}}_{t_{u+1}})]$.
\begin{proposition}\label{prop:moments_BS_model}
Assume that $X_{t}$ is a stochastic process with $X_{t_{u+1}}\vert X_{t_{u}}=x_{k}\sim\mathcal{N}(x_{k}+\Delta t\:\!\:\!\mu,\Delta t\sigma^{2})$ with $\Delta t=t_{u+1}-t_{u}$. Then the conditional moments can be written as
\begin{align*}
&\EE[p_j(X_{t_{u+1}})\vert X_{t_u}=x]=\EE[T_{j}(Y)\1_{[-1,1]}(Y)]\\
&Y\sim\mathcal{N}\Big(1-2\frac{\ol{x}-x}{\ol{x}-\ul{x}}+\frac{2}{\ol{x}-\ul{x}}\Delta t\:\!\:\!\mu, \big(\frac{2}{\ol{x}-\ul{x}}\big)^{2}\Delta t\sigma^{2}\Big).
\end{align*}
\end{proposition}
\begin{proof}
From the properties of a Brownian motion with drift follows
\begin{align*}
\EE[p_j(X_{t_{u+1}})\vert X_{t_u}=x]=\EE[p_j(x+(X_{t_{u+1}}-X_{t_{u}}))]=\EE[p_j(x+X_{\Delta t})].
\end{align*}
The definition of $p_{j}$ and the inverse of the linear transformation $\tau_{[\ul{x},\ol{x}]}$ yield
\begin{align*}
\EE[p_j(x+X_{\Delta t})]&=\EE[T_j(\tau^{-1}_{[\ul{x},\ol{x}]}(x+X_{\Delta t})\1_{[\ul{x},\ol{x}]}(x+X_{\Delta t})]\\
&=\EE[T_{j}(1-2\frac{\ol{x}-(x+X_{\Delta t})}{\ol{x}-\ul{x}})\1_{[\ul{x},\ol{x}]}(x+X_{\Delta t})]\\
&=\EE[T_{j}(1-2\frac{\ol{x}-x}{\ol{x}-\ul{x}}+\frac{2}{\ol{x}-\ul{x}}X_{\Delta t})\1_{[\ul{x},\ol{x}]}(x+X_{\Delta t})]\\
&=\EE[T_{j}(Y)\1_{[-1,1]}(Y)]
\end{align*}
with $Y$ defined as
\begin{align*}
Y=1-2\frac{\ol{x}-x}{\ol{x}-\ul{x}}+\frac{2}{\ol{x}-\ul{x}}X_{\Delta t}
\end{align*}
and we used that for a linear transformation holds
\begin{align*}
\1_{[\ul{x},\ol{x}]}(x+X_{\Delta t})]=\1_{[\tau^{-1}_{[\ul{x},\ol{x}]}(\ul{x}),\tau^{-1}_{[\ul{x},\ol{x}]}(\ol{x})]}(\tau^{-1}_{[\ul{x},\ol{x}]}(x+X_{\Delta t}))=\1_{[-1,1]}(Y).
\end{align*}
The properties of a normally distributed variable yields our claim.
\end{proof}

\begin{proposition}
Let $Y\sim\mathcal{N}(\mu,\sigma^{2})$ be a normally distributed random variable with density $f$ and distribution function $F$. The truncated generalized moments $\mu_{j}=\EE[T_{j}(Y)\1_{[-1,1]}(Y)]$ are recursively defined by
\begin{align*}
\mu_{n+1}=2\mu\mu_{n} - 2\sigma^{2}\big(f(1)-f(-1)T_{n}(-1)-2n\sum_{j=0}^{n-1}{}^{'}\mu_{j}\1_{(n+j)\bmod 2=1}\big)-\mu_{n-1}
\end{align*}
for $n\geq 1$ and starting values $\mu_{0}=F(1)-F(-1)$, $\mu_{1}=\mu\mu_{0}-\sigma^{2}(f(1)-f(-1)$ and where $\sum{}^{'}$ indicates that the first term is multiplied with $1/2$.
\end{proposition}
\begin{proof}
The proof can be found in the appendix.
\end{proof}
For a large model class for which the underlying process is conditionally normally distributed or can be approximated by such a process, the conditional moments can thus be efficiently computed by an analytic formula.\\

If the underlying process is not normally distributed numerical approximation techniques come into play. \cite{GlauMahlstedtPoetz2019} give an overview of different approaches which can be used to calculate the conditional expectations. For example numerical quadrature techniques using the density or characteristic function of the process or with the help of Monte Carlo simulations. The possibility to use different approaches gives us the flexibility to apply the method in a variety of models.\\

When we use an equidistant time stepping $t_{u+1}-t_{u}=\Delta t$ the problem can be further simplified. Assuming
\begin{align}\label{eq_stationarity_assumption}
\EE^{\mathbb{Q}}\left[p_{j}(X_{t_{u+1}})\vert X_{t_{u}}=x_{k}\right]
=\EE^{\mathbb{Q}}\left[p_{j}(X_{\Delta t})\vert X_{0}=x_{k}\right],
\end{align}
the pre-computation step becomes independent of the maturity $T$ and the number of time steps $n$. We only have to simulate the underlying at $\Delta t$. Equation \eqref{eq_stationarity_assumption} holds if the process $(X_{t})_{0\leq t\leq T}$ has stationary increments.\\ 

\textbf{Delta and Gamma as by-product of the method}\\
Generally, the efficiency of the method allows a fast computation of sensitivities via bump and re-run. For Delta and Gamma the polynomial structure of the Chebyshev approximation allows for a direct computation without re-running the time-stepping. Instead we only need to differentiate a polynomial. For Delta we obtain 
\begin{align*}
\frac{\partial V_{t}}{\partial x}(x)\approx\sum_{j=0}^{N}c_{j}^{t}\frac{\partial p_{j}}{\partial x}(x),
\end{align*} 
which is again a polynomial with degree $N-1$ and for Gamma we obtain
\begin{align*}
\frac{\partial^{2} V_{t}}{\partial x^{2}}(x)\approx\sum_{j=0}^{N}c_{j}^{t}\frac{\partial^{2} p_{j}}{\partial x^{2}}(x),
\end{align*} 
a polynomial of degree $N-2$. These formulas can be used to calculate the derivative of $V_{t}$ with respect to $x_{t}$. For the derivative w.r.t. $x_{0}$ we obtain via chain rule $\partial V_{t}/\partial x_{0}=\partial V_{t}/\partial x_{t}\cdot\partial x_{t}/\partial x_{0}$.\\

\textbf{Several options on one underlying:}\\
The structure of the dynamic Chebyshev algorithm for exposure calculation exhibits additional benefits for the complex derivative portfolios. For instance, consider non-directional strategies and structured products that offer different levels of capital protection or enhanced exposure. They are typically constructed from a combination of European options, with different strikes and maturities, together with Bermudan options and barrier options. Such structures are essentially a portfolio of derivatives on the same underlying asset, and in this case, the pricing and exposure calculation can be simplified by choosing the same interpolation domain. First, we only need to compute the conditional moments once and  then we can use them for all options. Second, we require less computation in the exposure calculation. Assume we have two options and we are in the time stepping of the Dynamic Chebyshev algorithm at step $t_{u}$. We have two Chebyshev approximations $\widehat{V}^{1}_{t_{u}}=\sum c^{1}_{j}(t_u)p_j$ and $\widehat{V}^{2}_{t_{u}}=\sum c^{2}_{j}(t_u)p_j$. For the exposure calculation we need to compute $\widehat{V}^{1/2}_{t_u}(X_{t_{u}}^{i})=\sum c^{1/2}_{j}(t_u)p_j(X_{t_{u}}^{i})$ for all risk factors $i=1,\ldots,M$. Hence the evaluation of the Chebyshev polynomials $p_{j}$ at the risk factors $X_{t_{u}}^{i}$ is the same and has only to be done once. In summary, with low additional effort, we can calculate the exposure of several options on one underlying. 

\subsection{Implementational aspects of the DC method for exposure calculation}
In this section, we discuss several implementational aspects which can help to achieve a high performance.\\

\textbf{Choice of interpolation domain:}\\
The choice of a suitable interpolation domain is an important step to ensure a high efficiency of the method. In general the choice of the domain is a trade-off between speed (small domain, low number of nodal points) and accuracy (larger domain, more nodal points). In general, we want to choose the interpolation domain in dependence of the underlying distribution. A suitable choice for the lower boundary is the $p$-quantile for a small $p$ (e.g. $10^{-4}$, $10^{-5}$) and similarly the $1-p$ quantile as the upper boundary. If the underlying risk factor is normally distributed with $X_{T}bv\sim\mathcal{N}(\mu_{T},\sigma_{T}^{2})$ we define the interval
\begin{align*}
[\ul{x},\ol{x}]=[\mu_{T} -k\cdot\sigma,\mu + k\cdot\sigma]
\end{align*}
for some $k>0$. For most applications $k=4$ or $k=5$ is sufficient. 

If $X_{t_{u}}^{i}\notin [\ul{x},\ol{x}]$ one can explore additional knowledge of the specific product. First we consider a Bermudan put option. Here we know that the value of the option converges towards zero if the (log-) price of the underlying goes to infinity. The upper bound $\ol{x}$ is therefore no problem and if we have a risk factor with $X^{i}_{t_{u}}>\ol{x}$ we can simply set $V_{t_{u}}(X^{i}_{t_{u}})=0$. For very low values of $x$ the option is always exercised and thus we set $V_{t_{u}}(X^{i}_{t_{u}})=g(X^{i}_{t_{u}})$ if $X^{i}_{t_{u}}<\ul{x}$.

For an European call or put option we can use the Call-Put parity $C_{t}(x)-P_{t}(x)=e^{x}-e^{-r(T-t)}K$ to find a suitable interpolation domain. The price of a call option converges towards zero for small $x$ and towards $e^{x}-e^{-r(T-t)}K$ for large $x$. We choose $\ul{x}$, $\ol{x}$ such that $C_{t}(\ul{x})$ and $P_{t}(\ol{x})$ are sufficiently small. Then we can set $V_{t_{u}}(X^{i}_{t_{u}})=0$ for $X^{i}_{t_{u}}<\ul{x}$ and $V_{t_{u}}(X^{i}_{t_{u}})=e^{X^{i}_{t_{u}}}-e^{-r(T-t)}K$ for $X^{i}_{t_{u}}>\ol{x}$.

As our last example we consider an up-and-out call option with barrier $b$. Here $b$ is the logical upper bound of the interpolation domain and for $\ul{x}$ we proceed similarly to the European call option case.\\

\textbf{Smoothing:}\\
If the payoff of the option has a kink or discontinuity the approximation with Chebyshev polynomials is not efficient. In this case we can modify the algorithm and improve convergence by a "smoothing" of the first time step. We can exploit that the continuation value at $t_{n-1}$ is exactly the value of a European option with duration $\Delta t=t_{n}-t_{n-1}$, i.e.
\begin{align}\label{eq:payoff_smoothing}
V_{t_{n-1}}(x)=\max\{g(x),P^{EU}(x)\} \quad \text{with} \quad P^{EU}(x)=\EE^{\mathbb{Q}}[g(X_{t_{n}})\vert X_{t_{n-1}}=x].
\end{align}
Often, it is more efficient to compute directly the European option price $\EE^{\mathbb{Q}}[g(X_{t_{n}})\vert X_{t_{n-1}}=x_{k}]$ at the nodal points $x_{k}$, $k=0,\ldots,N$. Hence, there is no interpolation error in the first step and we start with the interpolation of the (smooth) function $V_{t_{n-1}}$. We use this technique for all our numerical experiments. The influence of this modification on the error decay is investigated in \cite{GlauMahlstedtPoetz2019}.\\

\textbf{Splitting of the interpolation domain:}\\
If the value function is not analytic or the interpolation domain is large the degree of the Chebyshev domain $N$ increases. This makes the evaluation of the closed form approximation in each time step more costly. In this case it is often beneficial to split the domain into two subdomains and interpolate on each of the subdomains. On each of the subdomains, we require significantly less nodal points and the interpolation becomes more efficient. A suitable choice for the splitting point is the strike $K$ of the option (or $k=\log(K)$) for an equity option. By doing so the smoothing mentioned in the previous section is no longer required.\\

The splitting changes the dynamic Chebyshev algorithm in the following way. Assume the value function at time point $t_{u+1}$ is approximated by two Chebyshev interpolants, i.e. $V_{t_{u+1}}=\widehat{V}^{1}_{t_{u+1}}\1_{[\ul{x},k]} + \widehat{V}^{2}_{t_{u+1}}\1_{(k,\ol{x}]}$. In order to approximate $V_{t_{u}}$ we require the nodal values for two sets of nodal points $x_{k}^{1}$ and $x_{k}^{2}$ given by
\begin{align*}
V_{t_{u}}(x_{k}^{1})&=f\Big(g(t_{u},x_{k}^{1}),\sum_{j=0}^{N_{1}}c_{j}^{1}(t_{u+1})\EE^{\mathbb{Q}}\left[p_{j}(X_{t_{u+1}})\1_{[\ul{x},k]}\vert X_{t_{u}}=x_{k}^{1}\right]\\
&\quad + \sum_{j=0}^{N_{2}}c_{j}^{2}(t_{u+1})\EE^{\mathbb{Q}}\left[p_{j}(X_{t_{u+1}})\1_{(k,\ol{x}]}\vert X_{t_{u}}=x_{k}^{1}\right]\Big)
\end{align*}
and the equivalent expression for the values $V_{t_{u}}(x_{k}^{2})$. Hence, in the pre-computation step we calculate four different sets of conditional expectations for polynomials $p_{j}\1_{[\ul{x},k]}$ and $p_{j}\1_{(k,\ol{x}]}$ and starting values $x_{k}^{1}$ and $x_{k}^{2}$. In comparison to a Chebyshev interpolation on the whole domain $[\ul{x},\ol{x}]$ with $N$ points we can choose a lower $N_{1}$, $N_{2}$. If we set $N_{1}=N_{2}=N/2$, the number of conditional expectations which we have to compute in the pre-computation step is exactly the same. For the exposure calculation we need to divide the paths $X_{t_{u}}$ into the ones below the splitting point and the ones above the splitting point.

\section{Numerical experiments}
In this section, we investigate the dynamic Chebyshev method numerically by calculating the credit exposure profiles of European and path-dependent equity options and a Bermudan swaption. We analyse the accuracy of the exposure profiles produced by the dynamic Chebyshev method by comparing them to a full re-evaluation. Then we investigate the method's performance and compare it to the popular LSM approach. Moreover, we check the influence of the proposed splitting of the domain on the method's performance.

\subsection{Description of the experiments}
For the numerical experiments we consider four different products: A European put option and an up-and-out barrier call option in the Black-Scholes model, a Bermudan put option in the Merton jump diffusion model and a Bermudan receiver swaption in the Hull-White short rate model. In the Black-Scholes and the Hull-White model the risk factor is normally distributed and we can use the analytic formula for the conditional expectations of the Chebyshev polynomials.

We compute the expected exposure $EE_{t}^{price}$ under the pricing measure $\mathbb{Q}$ and the expected exposure $EE_{t}^{risk}$ under the real-world measure $\mathbb{P}$ as well as the potential future exposures $PFE_{t}^{price}$ and $PFE_{t}^{risk}$ under both measures. For the calculation of the exposure measures we use $50000$ and $150000$ simulation paths of the underlying risk factors and a time discretization of $50$ time steps per year. The relatively high number of simulation paths is needed to obtain a stable estimate of the PFE over the lifetime of the derivative. Since the PFE is a tail measure it is more sensitive to the number of simulations than the expected exposure.

We run the dynamic Chebyshev method for a different number of nodal points $N$ and the dynamic Chebyshev with splitting approach with $N_{1}=N_{2}=N/2$ nodal points. For the LSM we use the monomials up to degree $5$ plus the payoff of the product as basis functions in the regression. The pricing is done using $150000$ path of the underlying risk factors and then we use a second set of paths for the calculation of the exposure. Using two different sets of paths for pricing and exposure calculation reduces the bias of the LSM. See \cite{KarlssonJainOosterlee2016} for a description on how to use the LSM approach to calculate credit exposures under the pricing and the real-world measure. In our implementation of the LSM approach for exposure calculation we use $7$ basis functions for the European and Bermudan equity options, $8$ for the barrier option and $5$ basis function for the Bermudan swaptions. For the pricing we use a separate set of $150000$ simulation paths of the underlying risk factor.

For the experiments we introduce the following three asset price models and explain how we compute the corresponding generalized moments.\\

\textbf{The Black-Scholes model}\\
In the classical model of \cite{BlackScholes1973} the stock price process is modelled by the SDE
\begin{align*}
\text{d}S_{t}=\mu S_{t}\text{d}t + \sigma S_{t}\text{d}W_{t}.
\end{align*}
with drift $\mu$ and volatility $\sigma>0$ under the real-world measure $\mathbb{P}$. Under the pricing measure $\mathbb{Q}$ the drift equals $r$. Exploiting the fact that the log-returns $X_{t}=\log(S_{t}/S_{0})$ are normally distributed we can use the analytic formula for the generalized moments $\Gamma_{k,j}$. As model parameter we fix  volatility $\sigma=0.25$, real-world drift $\mu=0.1$, interest rate $r=0.03$ and initial stock price $S_{0}=100$.\\

\textbf{The Merton jump diffusion model}\\
The jump diffusion model introduced by \cite{Merton1976} adds jumps to the classical Black-Scholes model. The log-returns follow a jump diffusion with volatility $\sigma$ and added jumps arriving at rate $\lambda>0$ with normal distributed jump sizes according to $\mathcal{N}(\alpha,\beta^{2})$. The stock price under $\mathbb{P}$ is modelled by the SDE
\begin{align*}
\text{d}S_{t}=\mu S_{t}\text{d}t + \sigma S_{t}\text{d}W_{t} + \text{d}J_{t}
\end{align*}
for a compound Poisson process $J_{t}$ with rate $\lambda$. The characteristic function of the log-returns $X_{t}=\log(S_{t}/S_0)$ under the pricing measure $\mathbb{Q}$ is given by
\begin{align*}
\varphi(z)=exp\left(t\left(ibz - \frac{\sigma^{2}}{2}z^{2} + \lambda\left(e^{iz\alpha - \frac{\beta^{2}}{2}z^{2}}-1\right)\right)\right)
\end{align*} 
with risk-neutral drift
\begin{align*}
b=r-\frac{\sigma^{2}}{2}-\lambda\left(e^{\alpha+\frac{\beta^{2}}{2}}-1\right).
\end{align*}
In our experiments we calculate the conditional expectations $\Gamma_{k,j}$ using numerical integration and the Fourier transforms of the Chebyshev polynomials along with the characteristic function of $X_t$. We fix the parameters 
\[\sigma=0.25,\quad \alpha=-0.5,\quad \beta=0.4,\quad \lambda=0.4\quad r=0.03\quad \text{and} \quad  \mu=0.1\]
and initial stock price $S_{0}=100$.\\

\textbf{The Hull-White model}\\
The Hull-White model as described in Chaper 3.3 of \cite{BrigoMercurio2007} is a short rate model where the rate process $(r_{t})_{t\geq 0}$ is a mean reverting Ornstein–Uhlenbeck process described by the SDE
\begin{align*}
\text{d}r_{t}=(\theta(t) - ar_{t})\text{d}t + \sigma\text{d}W_{t}
\end{align*}
where the long term mean $\theta(t)$ can be fitted to the term structure of the market and the speed of mean reversion $a$ and the volatility $\sigma$ are constant. One can write $r_{r}=\alpha(t) + x_{t}$ for a deterministic function $\alpha(t)$ given by
\begin{align*}
\alpha(t)=f^{M}(0,t)+\frac{\sigma^{2}}{2a^{2}}(1-e^{-at})^{2}
\end{align*}
where $f^{M}(0,t)$ is the market forward rate for maturity $T$ obtained from market discount factors $P^{M}(0,T)$ via
\begin{align*}
f^{M}(0,T)=-\frac{\partial\ln(P^{M}(0,T))}{\partial T}.
\end{align*}
The process $(x_{t})_{t\geq0}$ is modelled by the SDE
\begin{align*}
\text{d}x_{t}=- ax_{t}\text{d}t + \sigma\text{d}W_{t} \quad x_{0}=0
\end{align*}
and $x_{t}\vert x_{s}=x_{0}$ is normally distributed with
\begin{align*}
\EE[x_{t}\vert x_{s}=x_{0}]=x_{0}e^{-a(t-s)},\quad\text{and}\quad\text{Var}[x_{t}\vert x_{s}=x_{0}]=\frac{\sigma^{2}}{2a}(1-e^{-2a(t-s)}).
\end{align*}
As parameters we fix $a_{q}=0.02$ and $\sigma_{q}=0.02$ under $\mathbb{Q}$ and $a_{p}=0.015$ and $\sigma_{p}=0.01$ under $\mathbb{P}$ and we assume a flat forward rate $f^{M}(0,t)=0.01$. All parameters are taken from \cite{FengJainKarlssonKandhaiOosterlee2016}.

\subsection{European option in the Black-Scholes model}
In this section, we calculate the expected exposure and the potential future exposure of a European put option in the Black-Scholes model. In this case, we have an analytic formula for the option price $V_{t}$ at any time point $t$ and we can investigate the accuracy of the dynamic Chebyshev method for exposure calculation. We consider an at-the-money option with strike $K=100$ and maturity $T=1$.\\

Figure \ref{fig:EE_PFE_BS_European_Analytic} shows the resulting exposure profiles and Table \ref{tab:EE_PFE_BS_European_Analytic_Reference} shows the values of the exposures at maturity. The expected exposure under the pricing measure is constant since
\begin{align*}
EE^{price}_{t}&=D(0,t)\EE^{\mathbb{Q}}[\max\{V_{t}(X_{t}),0\}]\\
&=D(0,t)\EE^{\mathbb{Q}}[D(t,T)\EE^{\mathbb{Q}}[g(X_{T})\vert X_{t}]]\\
&=D(0,T)\EE^{\mathbb{Q}}[g(X_{T})]=V_{0}
\end{align*}
for all European options. Under the real-world measure the positive drift of the underlying yields a decreasing exposure for a put option. The PFE increases under both measures since it is mainly driven by the diffusion term in the model.

We observe that the exposure profiles of the dynamic Chebyshev method and the true exposure profile are indistinguishable. In Table \ref{tab:EE_PFE_BS_European_Analytic} we see the corresponding relative errors for different Chebyshev N's. Here, the notation $DC_{N}$ refers to the dynamic Chebyshev method of degree $N$ and $DC_{N_{1},N_{2}}$ refers to the dynamic Chebyshev method with domain splitting of degree $N_{1}$ and $N_{2}$. The error is calculated as the maximum over the simulation period and displayed in relative terms with respect to the initial option price. For $N=128$ the error is below $10^{-4}$ for both quantities and under both measures. Already for $N=64$ nodal points or, if splitting is applied, $N_{1}=N_{2}=16$ the error is below $1\%$ in each cases. Figure \ref{fig:EE_PFE_Err_BS_European_LSM} shows the error in the exposure profiles over the option's lifetime of the dynamic Chebyshev method with $N=128$ and compares it with the LSM approach. Whereas the dynamic Chebyshev method is able to produce stable results, the LSM is only able to produces accurate prices at $t=0$ but adds additional simulation noise over the option's lifetime. For the PFE, the LSM has an relative error of nearly $2\%$.

Table \ref{tab:EE_PFE_BS_European_Analytic_Runtime} shows the corresponding runtimes for $M=50000$ and $M=150000$ simulation paths of the underlying risk factor. The runtimes of the dynamic Chebyshev method increases approximately linearly in $M$ and is in the same region as the runtime of the analytic pricer. Moreover, the measure under which the risk factors are simulated has no influence on the runtime of the method. The fact that the new numerical method is competitive in comparison to a analytic formula indicates a high efficiency of the approach. In in comparison to the LSM, the dynamic Chebyshev method for exposure calculation is as fast or faster and, as already seen, able to produce more accurate results. For example, the dynamic Chebyshev method with $N_{1}=N_{2}=32$ is in all cases more accurate than the LSM but also always faster.

\begin{figure}[H]
\begin{minipage}{.5\linewidth}
\centering
\subfloat[]{\includegraphics[scale=.48]{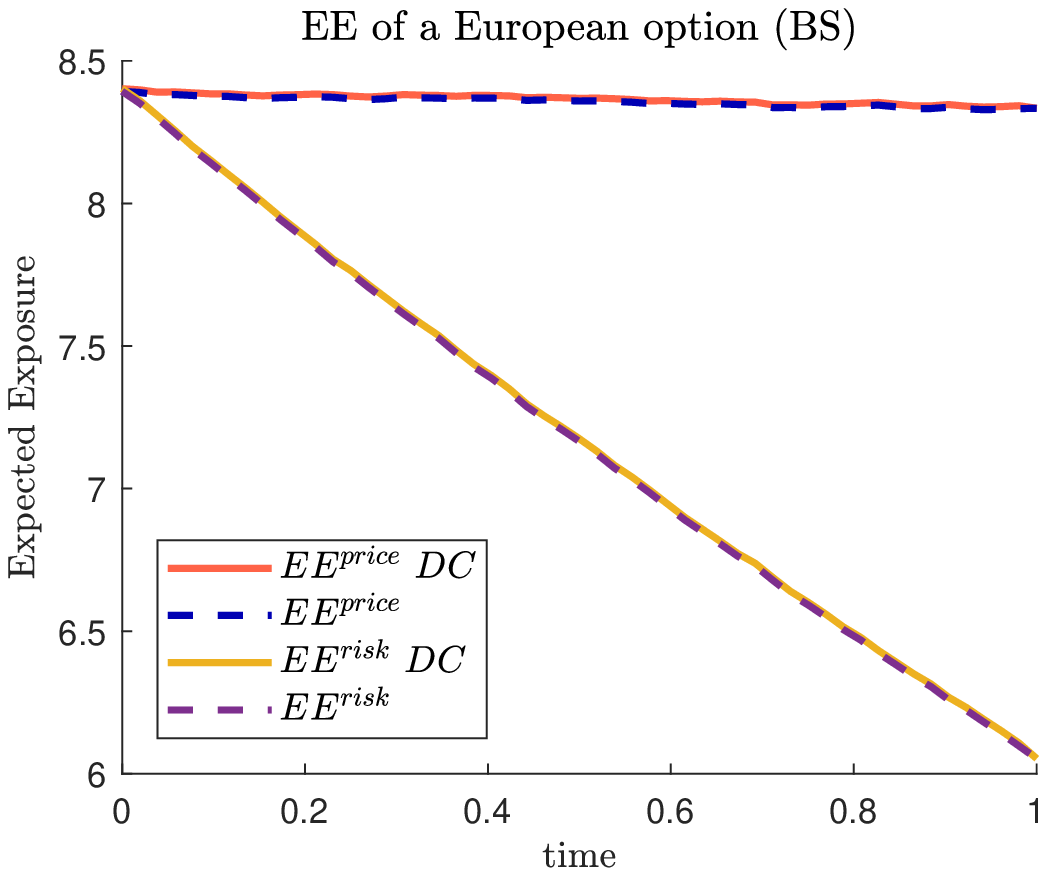}}
\end{minipage}%
\begin{minipage}{.5\linewidth}
\centering
\subfloat[]{\includegraphics[scale=.48]{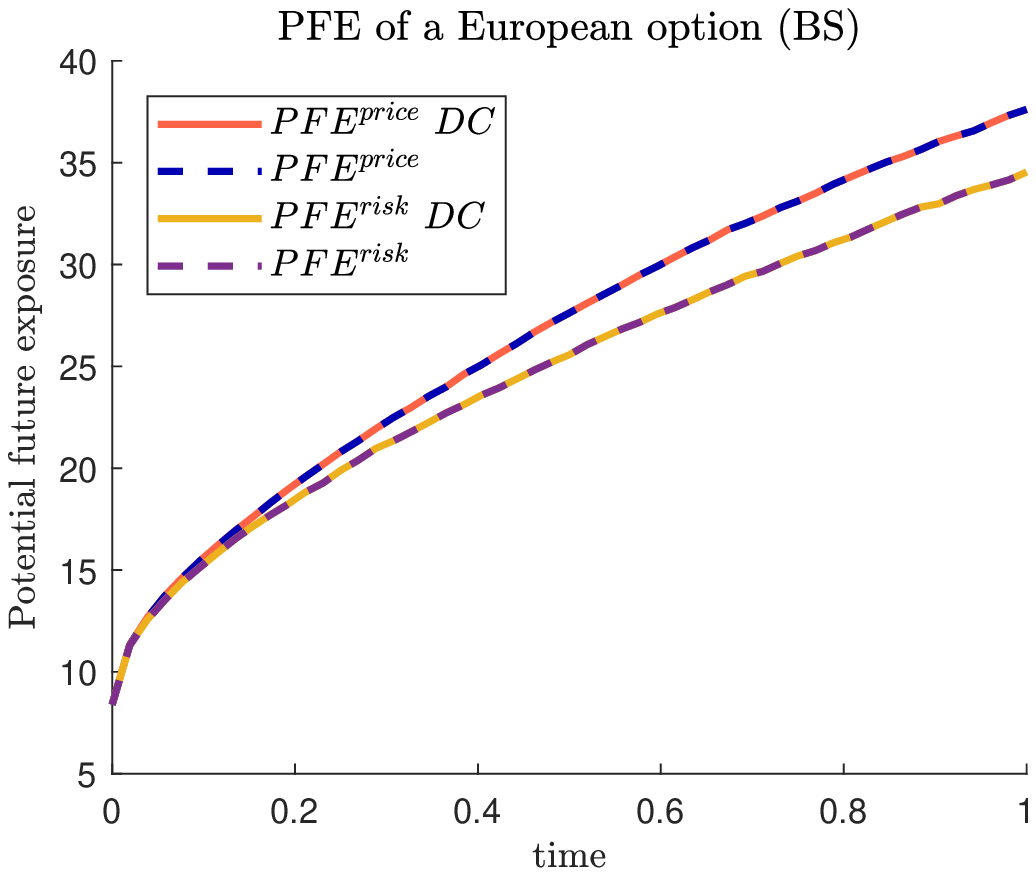}}
\end{minipage}
\caption{Expected exposure (left figure) and potential future exposure (right figure) of a European put option in the Black-Scholes model, calculated under the pricing measure $\mathbb{Q}$ and the real world measure $\mathbb{P}$ using $M=150000$ simulations.}
\label{fig:EE_PFE_BS_European_Analytic}
\end{figure}

\begin{figure}[H]
\begin{minipage}{.5\linewidth}
\centering
\subfloat[]{\includegraphics[scale=.48]{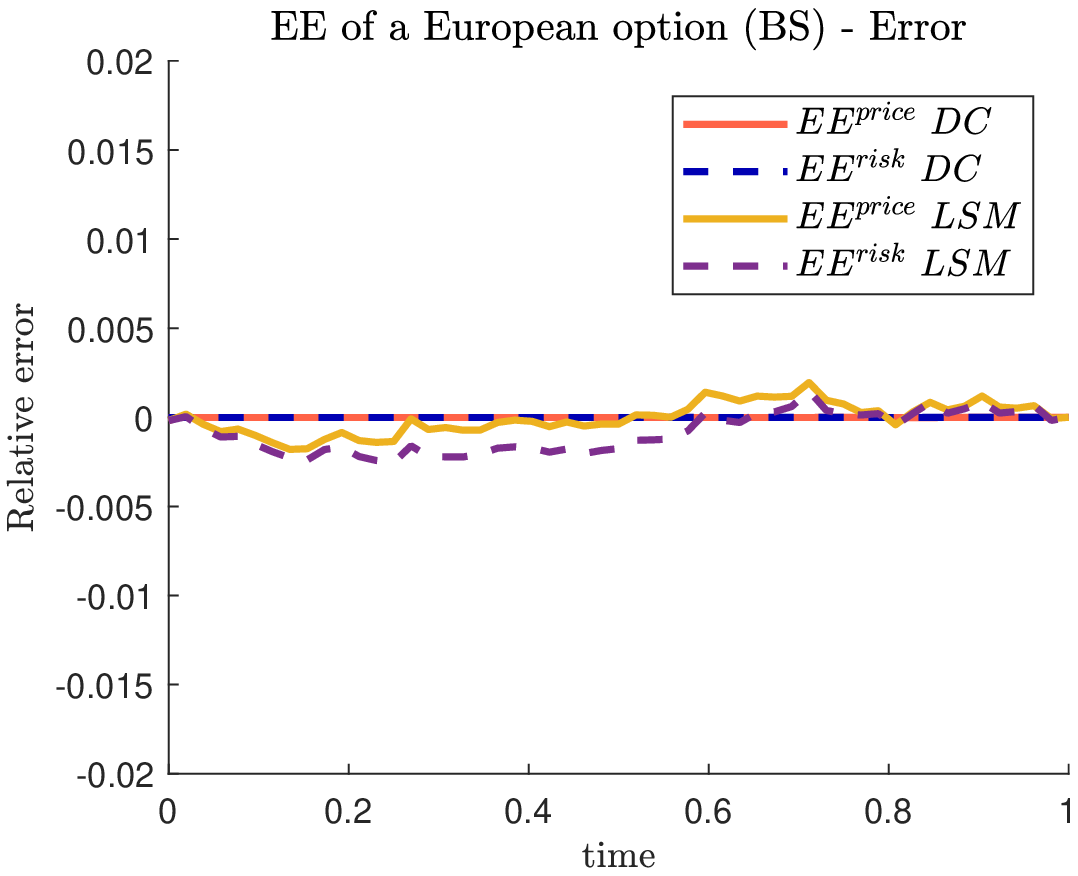}}
\end{minipage}%
\begin{minipage}{.5\linewidth}
\centering
\subfloat[]{\includegraphics[scale=.48]{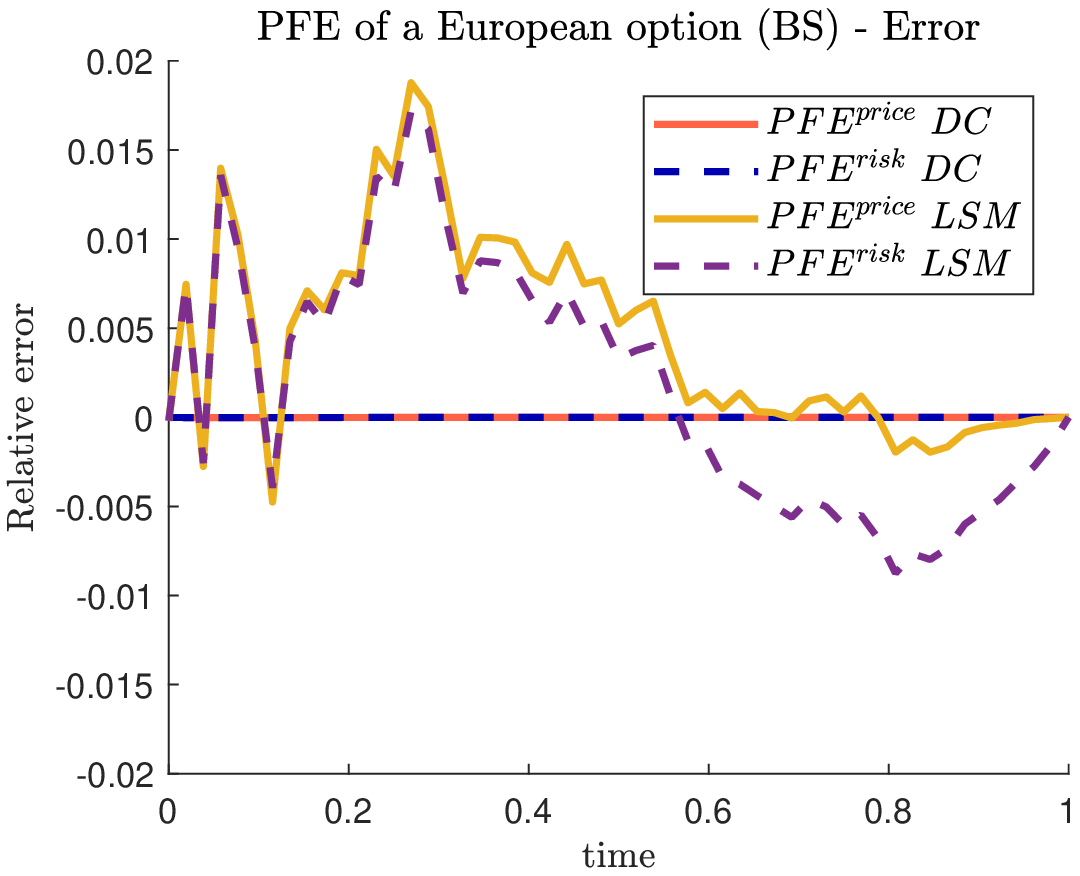}}
\end{minipage}
\caption{Relative error of the expected exposure (left figure) and potential future exposure (right figure) of a European put option in the Black-Scholes model of the dynamic Chebyshev method and the least-squares method. Calculated under the pricing measure $\mathbb{Q}$ and the real world measure $\mathbb{P}$ using $M=150000$ simulations.}
\label{fig:EE_PFE_Err_BS_European_LSM}
\end{figure}

\begin{table}[H]
\begin{center}
\begin{tabular}{ccccc}
\hline 
Price & $EE^{price}_{T}$ & $PFE^{price}$ & $EE^{risk}_{T}$ & $PFE^{risk}$\\
8.3930 & 8.3338 & 37.6163 & 6.0530 & 34.5426\\
\hline 
\end{tabular} 
\caption{Reference values for option price, EE and PFE of a European put option in the Black-Scholes  model using $M=150000$ simulations.}
\label{tab:EE_PFE_BS_European_Analytic_Reference} 
\end{center}
\end{table}

\begin{table}[H]
\begin{center}
\begin{tabular}{lccccc}
\hline 
 & Price & $EE^{price}$ & $PFE^{price}_{T}$ & $EE^{risk}$ & $PFE^{risk}_{T}$\\
\hline
$DC_{32}$ & 0.0020 & 0.0095 & 0.0020 & 0.0138 & 0.0020\\
$DC_{64}$ & 0.0011 & 0.0012 & 0.0011 & 0.0017 & 0.0011\\
$DC_{128}$ & 0.0000 & 0.0000 & 0.0000 & 0.0000 & 0.0000\\
$DC_{16, 16}$ split & 0.0003 & 0.0003 & 0.0020 & 0.0003 & 0.0020\\
$DC_{32, 32}$ split & 0.0000 & 0.0000 & 0.0011 & 0.0000 & 0.0011\\
$DC_{64, 64}$ split & 0.0000 & 0.0000 & 0.0000 & 0.0000 & 0.0000\\
LSM & 0.0002 & 0.0019 & 0.0188 & 0.0025 & 0.0172\\
\hline 
\end{tabular} 
\caption{Maximal relative error of option price, EE and PFE of a European put option in the Black-Scholes model for $M=150000$ simulations. Comparison of the dynamic Chebyshev approach for different $N$ with an analytic formula.}
\label{tab:EE_PFE_BS_European_Analytic} 
\end{center}
\end{table}

\begin{table}[H]
\begin{center}
\begin{tabular}{llcccccccc}
\hline 
& Sim. & $DC_{32}$ & $DC_{64}$ & $DC_{128}$ & $DC_{16,16}$ & $DC_{32,32}$ & $DC_{64,64}$ & LSM & BS\\  
\hline 
\multirow{2}{*}{$\mathbb{Q}$}
& $50k$ & 0.16s & 0.17s & 0.25s & 0.18s & 0.20s & 0.22s & 0.49s & 0.17s\\
& $150k$ & 0.53s & 0.58s & 0.76s & 0.59s & 0.62s & 0.67s & 0.77s & 0.63s\\
\multirow{2}{*}{$\mathbb{P}$}
& $50k$ & 0.17s & 0.19s & 0.26s & 0.20s & 0.21s & 0.24s & 0.49s & 0.18s\\
& $150k$ & 0.54s & 0.60s & 0.71s & 0.59s & 0.62s & 0.68s & 0.77s & 0.63s\\
\hline
\end{tabular} 
\caption{Runtimes of the exposure calculation using the dynamic Chebyshev method for different $N$. Comparison with the analytic Black-Scholes formula.}
\label{tab:EE_PFE_BS_European_Analytic_Runtime} 
\end{center}
\end{table}

Overall, the experiment confirms that the new approach is able to produce accurate credit exposure profiles both under the pricing measure $\mathbb{Q}$ and the real-world measure $\mathbb{P}$. Moreover, we have seen that the computations under the real-world measure are as fast as the computation under the pricing measure. For the European put option the runtimes were comparable to using the analytic Black-Scholes formula. Building on these very promising results we will investigate the performance for derivatives which are path-dependent and therefore in general more difficult to price.

\subsection{Barrier option in the Black-Scholes model}
In this section, we calculate the expected exposure and the potential future exposure of a discretely monitored up-and-out barrier call option in the Black-Scholes model. Due to the additional barrier the option becomes path-dependent and there is no longer an analytic solution. In order to compute reference prices we use the COS method provided in the benchmarking project of \cite{Benchop2015}. We consider an option with strike $K=100$, barrier $B=130$ and maturity $T=1$. We assume that the barrier option is discretely monitored and the monitoring dates coincide with dates for the exposure calculation.\\

Figure \ref{fig:EE_PFE_BS_Barrier_COS} shows the resulting exposure profiles and Table \ref{tab:EE_PFE_BS_Barrier_COS_Reference} shows the values of the exposures at maturity. The expected exposure under the pricing measure is constant over time which can be justified by the same arguments as for the European option. Under the real-world measure, the expected exposure increases slightly and for the PFE we observe also an increase. In comparison to the European option we see a slower increase in the beginning and a faster increase close to maturity. Here, we observe the effect of the barrier which means that an increase in the stock price also leads to a higher risk of triggering the barrier and a zero exposure afterwards. This effect is more problematic for longer a time to maturity. As for the European option, the exposure profiles of the dynamic Chebyshev method and the exposure profile of the full re-evaluation are indistinguishable. Figure \ref{fig:EE_PFE_Err_BS_Barrier_LSM} shows the error of the exposure profiles computed with the dynamic Chebyshev method for $N=64$ and with the least squares Monte Carlo approach. For the least squares Monte Carlo approach we added an additional basis function compared to the European version to better fit the barrier. In Table \ref{tab:EE_PFE_BS_Barrier_COS} we see the corresponding relative errors for different Chebyshev N's and the error of the LSM. The error is calculated as the maximum over the simulation period and displayed in relative terms. For $N=64$ the error is below $10^{-4}$ for both quantities. Figure \ref{fig:EE_PFE_Err_BS_Barrier_LSM} shows the relative error of the dynamic Chebyshev method and the LSM over the option's lifetime. We can again observe that a strong fluctuation in the error of the LSM and a stable and very low error for the dynamic Chebyshev method.

Table \ref{tab:EE_PFE_BS_Barrier_COS_Runtime} shows the corresponding runtimes for $M=50000$ and $M=150000$ simulation paths of the underlying risk factor. We observe that the dynamic Chebyshev method is more than $100$ times faster than doing a full-revaluation approach using an already competitive pricer. Compared to the LSM the dynamic Chebyshev method produces more accurate estimates while also being faster. The barrier yields a faster interpolation domain and therefore a lower number of interpolation nodes for the dynamic Chebyshev method. On the other side, the LSM does not profit from the barrier but we had to add an additional basis function to achieve a satisfying accuracy.

\begin{figure}[H]
\begin{minipage}{.5\linewidth}
\centering
\subfloat[]{\includegraphics[scale=.48]{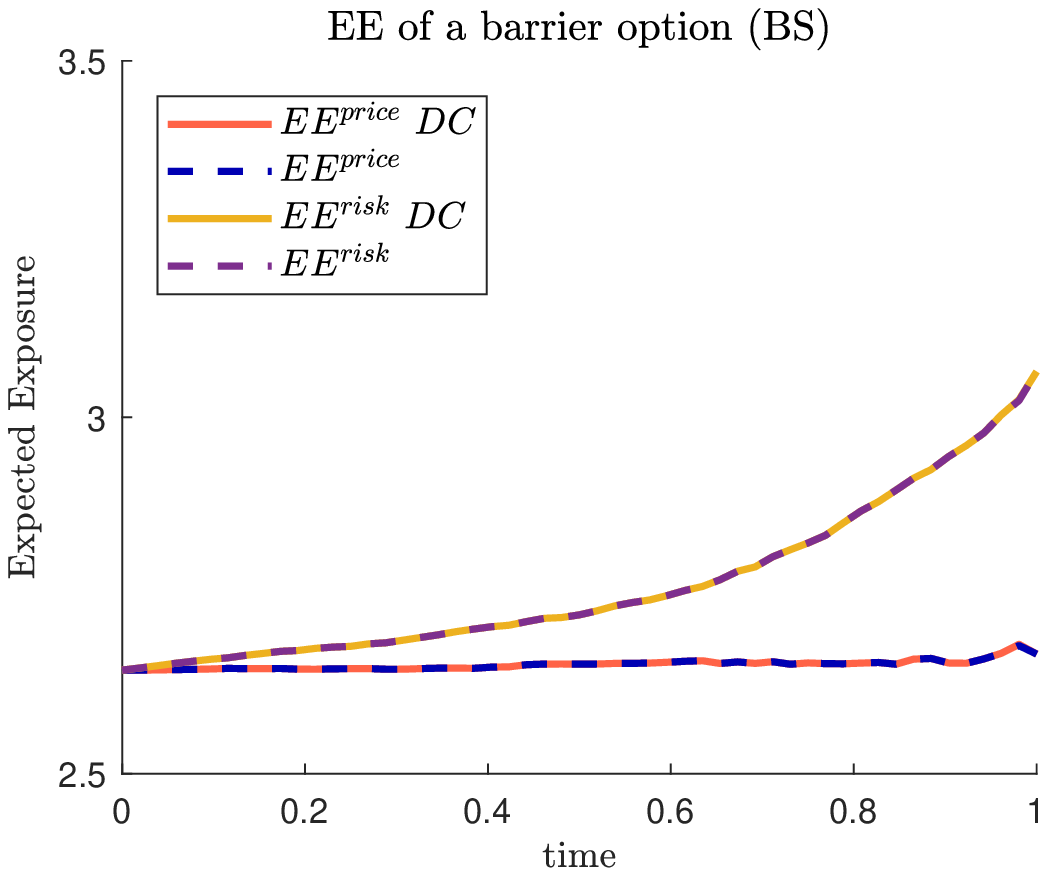}}
\end{minipage}%
\begin{minipage}{.5\linewidth}
\centering
\subfloat[]{\includegraphics[scale=.48]{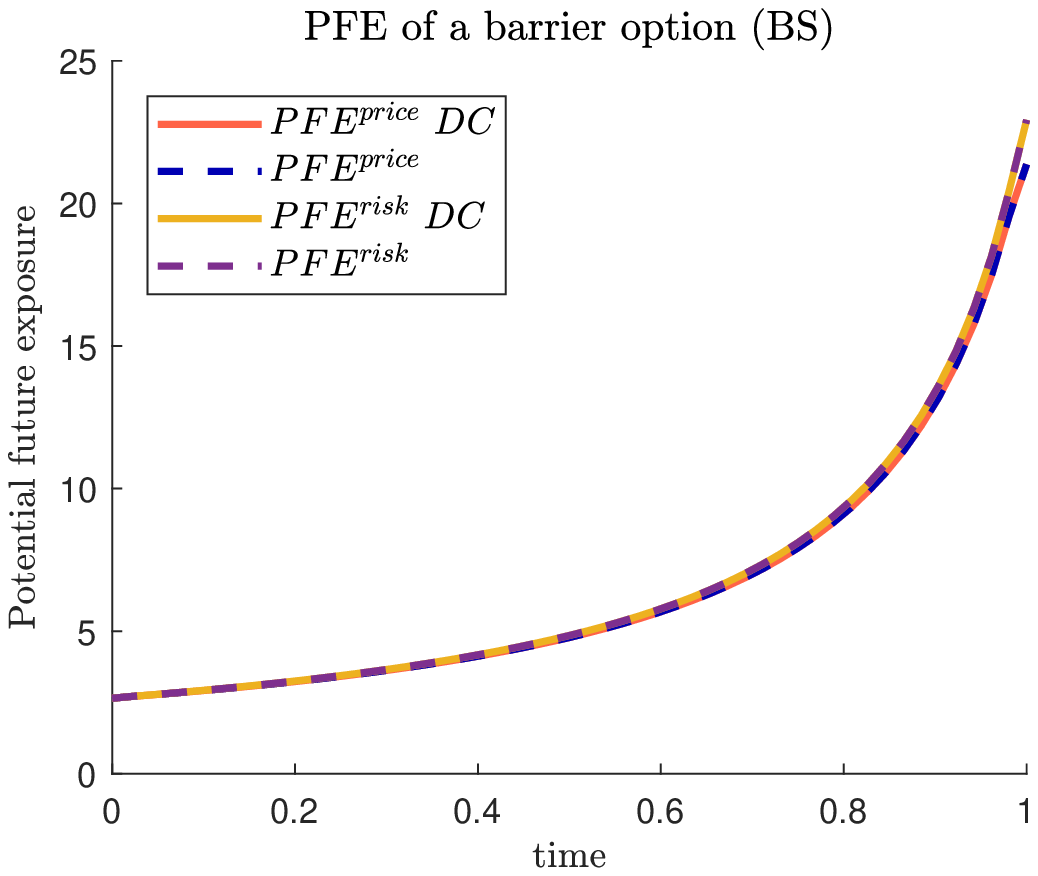}}
\end{minipage}
\caption{Expected exposure (left figure) and potential future exposure (right figure) of a barrier call option in the Black-Scholes model, calculated under the pricing measure $\mathbb{Q}$ and the real world measure $\mathbb{P}$ using $M=150000$ simulations.}
\label{fig:EE_PFE_BS_Barrier_COS}
\end{figure}

\begin{figure}[H]
\begin{minipage}{.5\linewidth}
\centering
\subfloat[]{\includegraphics[scale=.48]{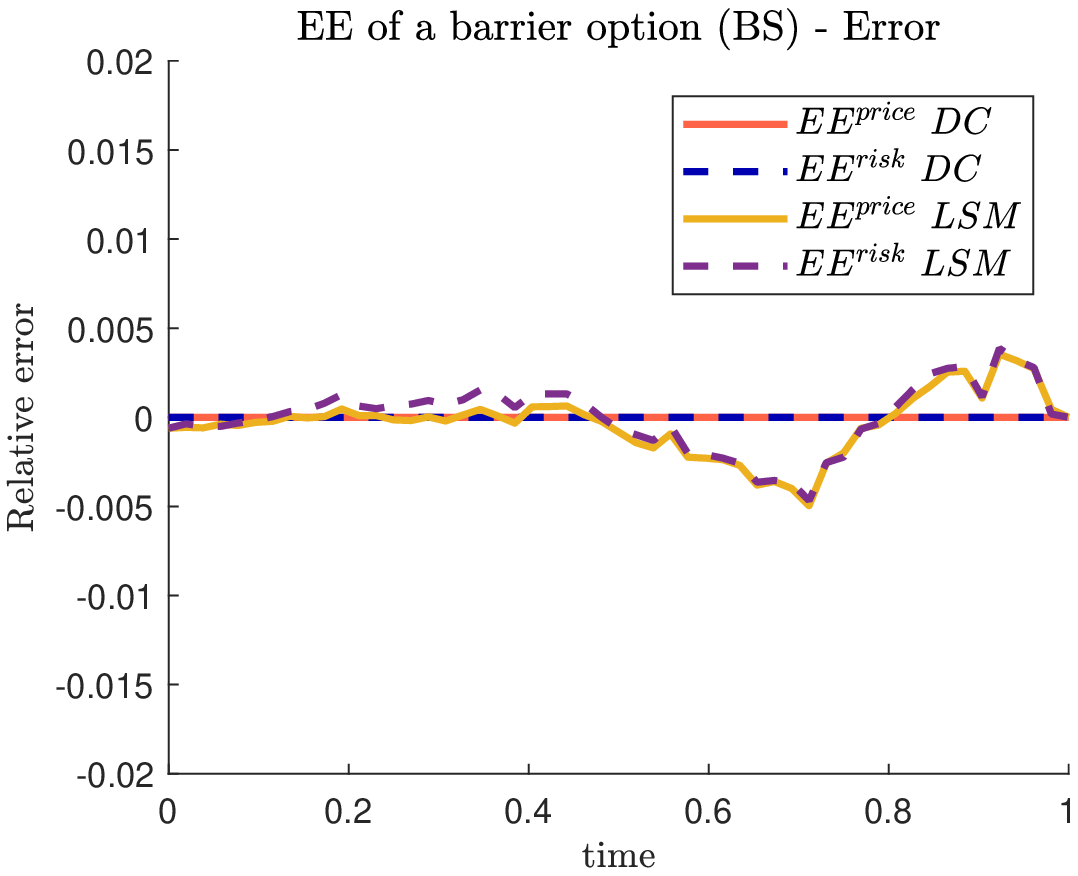}}
\end{minipage}%
\begin{minipage}{.5\linewidth}
\centering
\subfloat[]{\includegraphics[scale=.48]{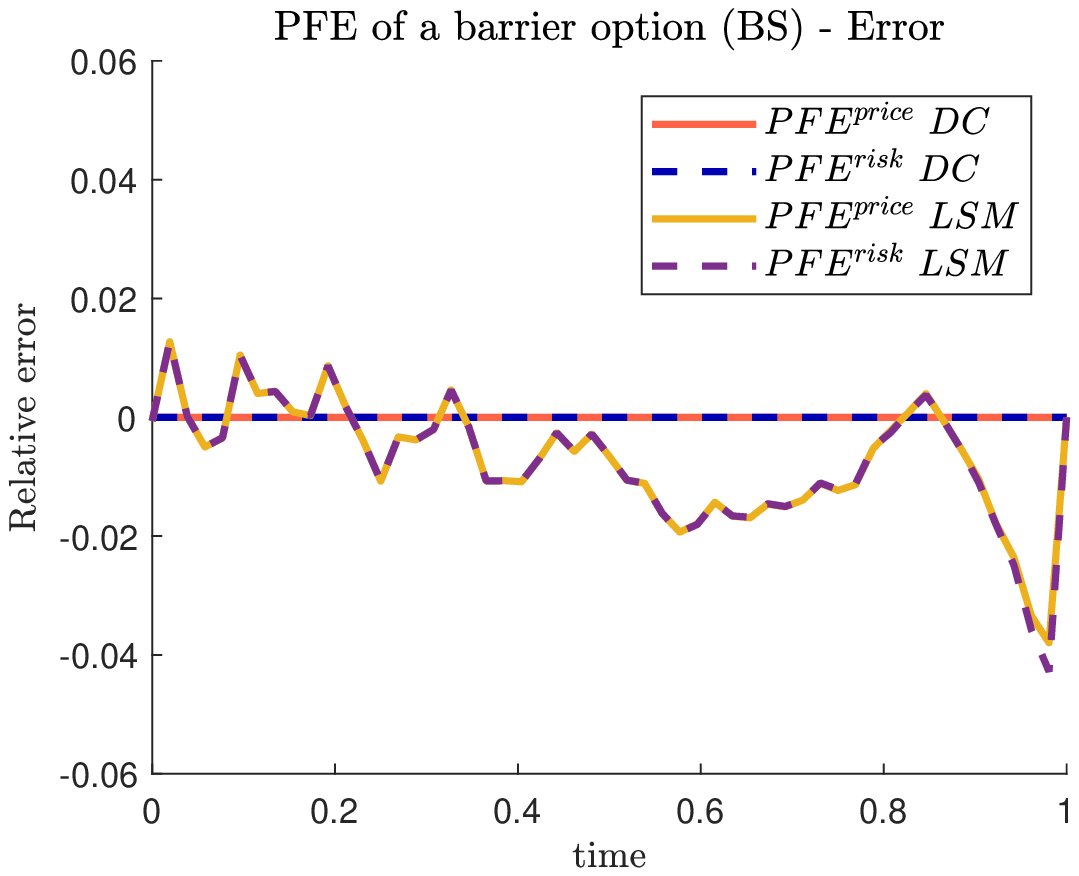}}
\end{minipage}
\caption{Relative error of the expected exposure (left figure) and potential future exposure (right figure) of a barrier up-and-out call option in the Black-Scholes model of the dynamic Chebyshev method and the least-squares method. Calculated under the pricing measure $\mathbb{Q}$ and the real world measure $\mathbb{P}$ using $M=150000$ simulations.}
\label{fig:EE_PFE_Err_BS_Barrier_LSM}
\end{figure}

\begin{table}[H]
\begin{center}
\begin{tabular}{ccccc}
\hline 
Price & $EE^{price}_{T}$ & $PFE^{price}$ & $EE^{risk}_{T}$ & $PFE^{risk}$\\
2.6453 & 2.6678 & 21.3718 & 3.0641 & 22.9297\\
\hline 
\end{tabular} 
\caption{Reference values for option price, EE and PFE of a barrier call option in the Black-Scholes  model using $M=150000$ simulations.}
\label{tab:EE_PFE_BS_Barrier_COS_Reference} 
\end{center}
\end{table}

\begin{table}[H]
\begin{center}
\begin{tabular}{lccccc}
\hline 
 & Price & $EE^{price}$ & $PFE^{price}_{T}$ & $EE^{risk}$ & $PFE^{risk}_{T}$\\
\hline
$DC_{16}$ & 0.0067 & 0.0172 & 0.0081 & 0.0133 & 0.0069\\
$DC_{32}$ & 0.0000 & 0.0006 & 0.0001 & 0.0004 & 0.0001\\
$DC_{64}$ & 0.0000 & 0.0000 & 0.0000 & 0.0000 & 0.0000\\
LSM & 0.0006 & 0.0050 & 0.0380 & 0.0047 & 0.0428\\
\hline 
\end{tabular} 
\caption{Maximal relative error of option price, EE and PFE of a barrier call option in the Black-Scholes model for $M=150000$ simulations. Comparison of the dynamic Chebyshev approach for different $N$ and the LSM approach with a full re-evaluation.}
\label{tab:EE_PFE_BS_Barrier_COS} 
\end{center}
\end{table}

\begin{table}[H]
\begin{center}
\begin{tabular}{llccccc}
\hline 
& Sim. & $DC_{16}$ & $DC_{32}$ & $DC_{64}$ & LSM & Full re-eval\\  
\hline 
\multirow{2}{*}{$\mathbb{Q}$}
& $50k$ & 0.17s & 0.17s & 0.19s & 1.28s & 24.2s\\
& $150k$ & 0.58s & 0.60s & 0.66s & 2.11s & 70.8s\\
\multirow{2}{*}{$\mathbb{P}$}
& $50k$ & 0.18s & 0.19s & 0.21s & 1.29s & 22.8s\\
& $150k$ & 0.59s & 0.62s & 0.68s & 2.03s & 66.5s\\
\hline
\end{tabular} 
\caption{Runtimes of the exposure calculation of a barrier call option using the dynamic Chebyshev method for different $N$. Comparison with a full re-evaluation using the COS method.}
\label{tab:EE_PFE_BS_Barrier_COS_Runtime} 
\end{center}
\end{table}

\subsection{Bermudan option in the Merton jump-diffusion model}
Here, we consider a Bermudan put option in the Merton jump-diffusion model. The early-exercise feature makes the option path-dependent and the jump component of the stock price model poses an additional computational challenge. Similar to the barrier option we can again use the COS method provided in the benchmarking project of \cite{Benchop2015} for the calculation of reference prices. We consider an option with strike $K=100$, maturity $T=1$ and we assume that the dates used for the exposure calculation are also the exercise dates of the option.\\

Figure \ref{fig:EE_PFE_MRT_Bermudan_COS} shows the resulting exposure profiles and Table \ref{tab:EE_PFE_MRT_Bermudan_COS_Reference} shows the values of the exposures at maturity. We observe a decreasing expected exposure under both measures due to the early exercise feature of the option. For the PFE we observe an increasing exposure in the beginning resulting from the diffusion term and a decreasing exposure afterwards. As for the European option, the exposure profiles of the dynamic Chebyshev method and the exposure profile of the full re-evaluation are indistinguishable. In Table \ref{tab:EE_PFE_MRT_Bermudan_COS} we see the corresponding relative errors for different Chebyshev N's. Due to the early-exercise feature, a higher number of nodal points is required for a similar accuracy of the option prices. Moreover, an exact estimation of the exercise barrier is critical for a correct estimation of the exposure. A miscalculation of the exercise barrier at time point $t_{u}$ does not only influence the exposure $EE_{t_{u}}$ but also the exposure at all future time points. For example, a barrier that is too low means that the option is exercised for too many paths and the exposure at future time points is underestimated. Figure \ref{fig:EE_PFE_Err_MRT_Bermudan_LSM} shows the relative error of the dynamic Chebyshev method and the LSM over the option's lifetime. We can see that the LSM struggles to provide accurate estimations in the tail and thus an accurate value for the PFE. In contrast, the dynamic Chebyshev method produces stable and accurate for both quantities.

Table \ref{tab:EE_PFE_MRT_Bermudan_COS_Runtime} shows the corresponding runtimes for $M=50000$ and $M=150000$ simulation paths of the underlying risk factor. We observe that the dynamic Chebyshev method is more than $100$ times faster than doing a full-revaluation. The comparison to the LSM shows again that the dynamic Chebyshev method is able to deliver both, more accurate results and faster runtimes.

\begin{figure}[H]
\begin{minipage}{.5\linewidth}
\centering
\subfloat[]{\includegraphics[scale=.48]{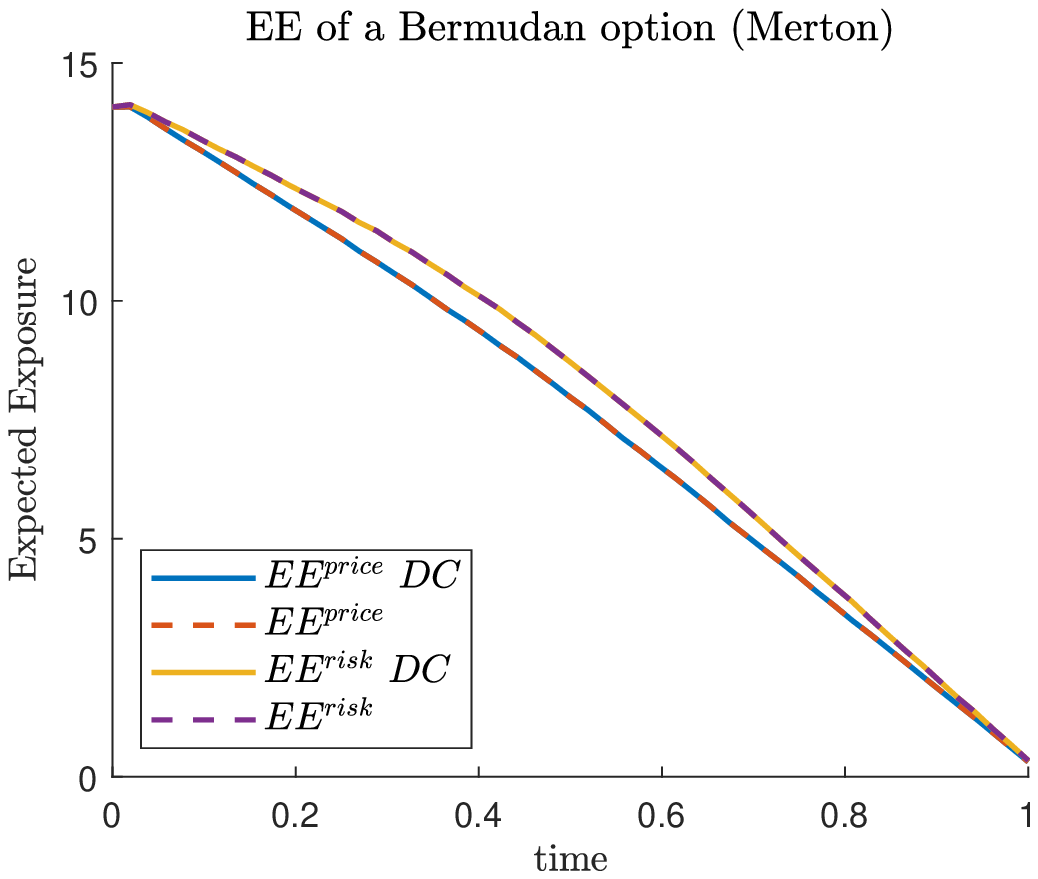}}
\end{minipage}%
\begin{minipage}{.5\linewidth}
\centering
\subfloat[]{\includegraphics[scale=.48]{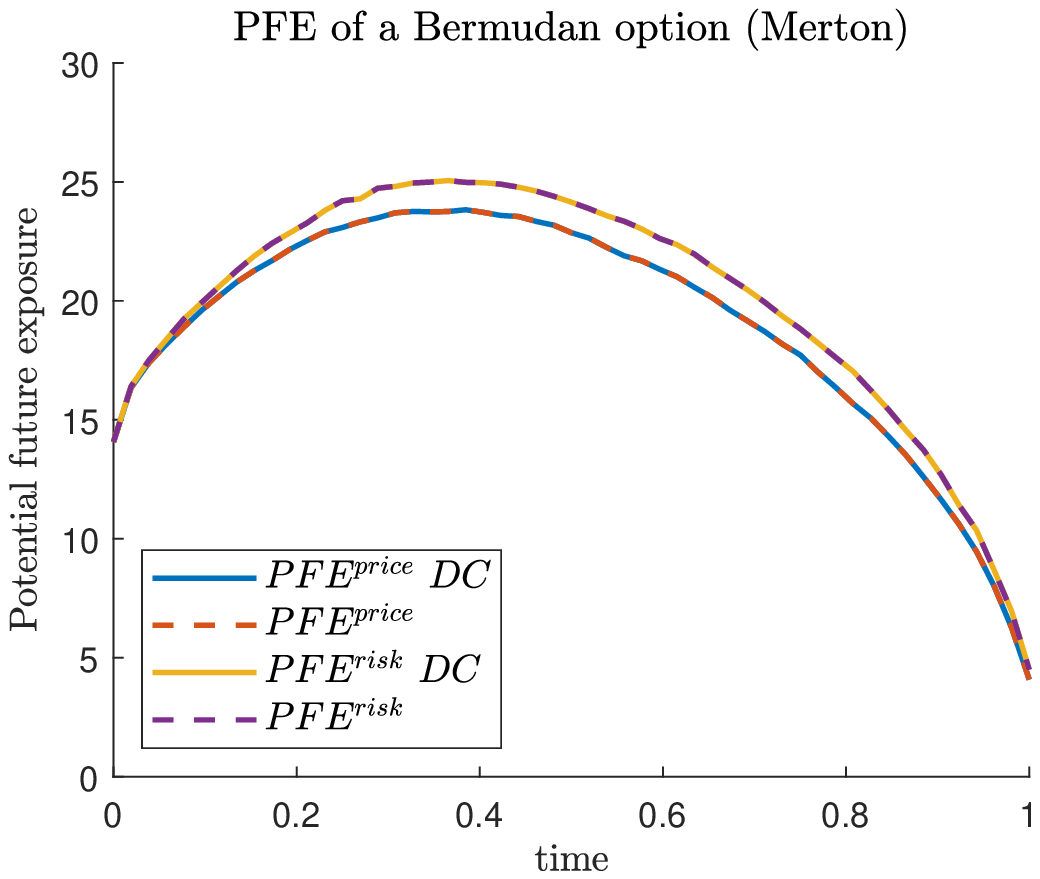}}
\end{minipage}
\caption{Expected exposure (left figure) and potential future exposure (right figure) of a Bermudan put option in the Merton jump-diffusion model, calculated under the pricing measure $\mathbb{Q}$ and the real world measure $\mathbb{P}$ using $M=150000$ simulations.}
\label{fig:EE_PFE_MRT_Bermudan_COS}
\end{figure}

\begin{figure}[H]
\begin{minipage}{.5\linewidth}
\centering
\subfloat[]{\includegraphics[scale=.48]{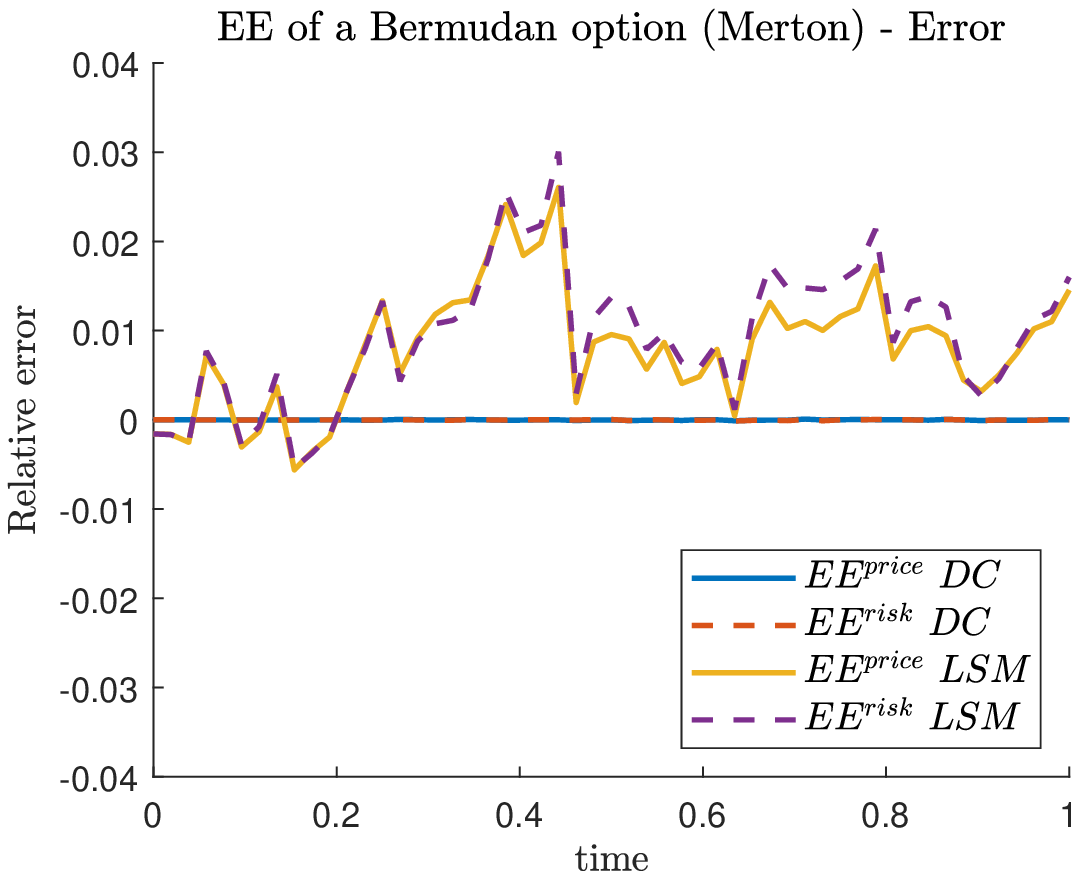}}
\end{minipage}%
\begin{minipage}{.5\linewidth}
\centering
\subfloat[]{\includegraphics[scale=.48]{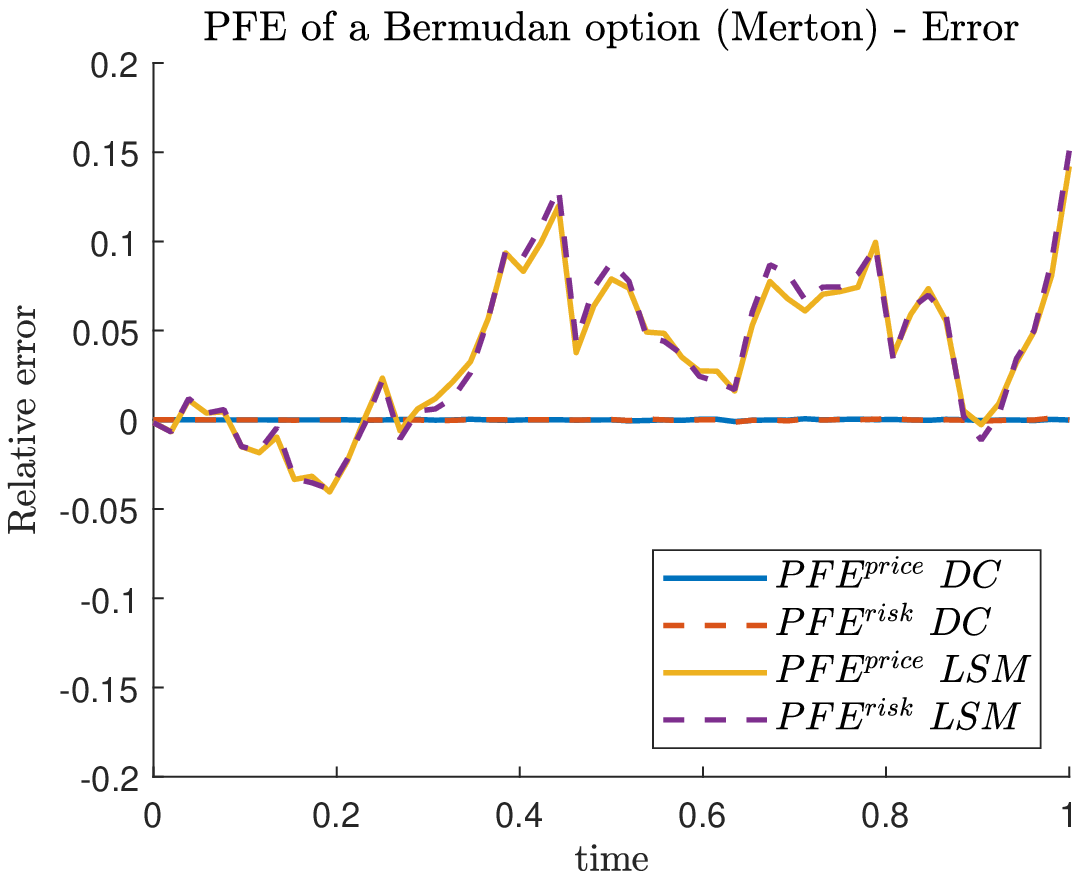}}
\end{minipage}
\caption{Relative error w.r.t. to the initial stock price of the expected exposure (left figure) and potential future exposure (right figure) of a Bermudan put option in the Merton jump-diffusion model of the dynamic Chebyshev method and the least-squares method. Calculated under the pricing measure $\mathbb{Q}$ and the real world measure $\mathbb{P}$ using $M=150000$ simulations.}
\label{fig:EE_PFE_Err_MRT_Bermudan_LSM}
\end{figure}

\begin{table}[H]
\begin{center}
\begin{tabular}{ccccc}
\hline 
Price & $EE^{price}_{T}$ & $PFE^{price}$ & $EE^{risk}_{T}$ & $PFE^{risk}$\\
14.0739 & 0.3144 & 4.1404 & 0.3601 & 4.6307\\
\hline 
\end{tabular} 
\caption{Reference values for option price, EE and PFE of a Bermudan put option in the Merton jump-diffusion model using $M=150000$ simulations.}
\label{tab:EE_PFE_MRT_Bermudan_COS_Reference} 
\end{center}
\end{table}

\begin{table}[H]
\begin{center}
\begin{tabular}{lccccc}
\hline 
 & Price & $EE^{price}$ & $PFE^{price}_{T}$ & $EE^{risk}$ & $PFE^{risk}_{T}$\\
\hline
$DC_{128}$ & 0.0001 & 0.0032 & 0.0433 & 0.0039 & 0.0568\\
$DC_{256}$ & 0.0000 & 0.0005 & 0.0028 & 0.0007 & 0.0032\\
$DC_{512}$ & 0.0000 & 0.0001 & 0.0007 & 0.0001 & 0.0007\\
$DC_{64, 64}$ & 0.0000 & 0.0020 & 0.0098 & 0.0023 & 0.0125\\
$DC_{128, 128}$ & 0.0000 & 0.0003 & 0.0018 & 0.0005 & 0.0028\\
$DC_{256, 256}$ & 0.0000 & 0.0001 & 0.0008 & 0.0001 & 0.0006\\
LSM & 0.0069 & 0.0277 & 0.1463 & 0.0331 & 0.1533\\
\hline 
\end{tabular} 
\caption{Maximal relative error w.r.t. to the initial stock price of option price, EE and PFE of a Bermudan put option in the Merton jump-diffusion model for $M=150000$ simulations. Comparison of the dynamic Chebyshev approach for different $N$ and the LSM approach with a full re-evaluation.}
\label{tab:EE_PFE_MRT_Bermudan_COS} 
\end{center}
\end{table}

\begin{table}[H]
\begin{center}
\begin{tabular}{llcccccccc}
\hline 
& Sim. & $DC_{128}$ & $DC_{256}$ & $DC_{512}$ & $DC_{64,64}$ & $DC_{128,128}$ & $DC_{256,256}$ & LSM & full\\  
\hline 
\multirow{2}{*}{$\mathbb{Q}$}
& $50k$ & 0.51s & 0.64s & 1.01s & 0.51s & 0.62s & 0.84s & 1.99s & 55s\\
& $150k$ & 1.45s & 1.70s & 2.39s & 1.47s & 1.64s & 1.99s & 3.20s & 162s\\
\multirow{2}{*}{$\mathbb{P}$}
& $50k$ & 0.46s & 0.61s & 0.93s & 0.48s & 0.59s & 0.83s & 1.97s & 57s\\
& $150k$ & 1.45s & 1.70s & 2.35s & 1.46s & 1.63s & 1.99s & 3.17s & 163s\\
\hline
\end{tabular} 
\caption{Runtimes of the exposure calculation of a barrier call option using the dynamic Chebyshev method for different $N$. Comparison with a full re-evaluation using the COS method.}
\label{tab:EE_PFE_MRT_Bermudan_COS_Runtime} 
\end{center}
\end{table}

\subsection{Bermudan swaption in the Hull-White model}
Here, we consider a Bermudan receiver swaption in the Hull-White model. Similar to the equity case, the early-exercise feature makes the option path-dependent and poses an additional computational challenge. Additionally, the payoff function is more complex and requires the pricing of a reciever swap. In the Hull-White model the prices of zero coupon bonds and swaps are still available analytically, see \cite{BrigoMercurio2007}. We consider a swaption with strike $K=0.01094$ and maturity $T$ which can be exercised yearly starting at $T_{1}=1$ and the swap ends terminates at $T+1$ and payments are also exchanged on a yearly basis. A detailed description of the pricing problem can be found in \cite{FengJainKarlssonKandhaiOosterlee2016}. From this paper we also obtain a reference price of $V_{0}=5.463$. We assume that the swaption is cash-settled and hence there is no credit exposure after the option is exercised.\\

Figure \ref{fig:EE_PFE_HW_Swaption} shows the resulting exposure profiles and Table \ref{tab:EE_PFE_HW_Swaption_Reference} shows the price and the values of the exposures at maturity. The price of the dynamic Chebyshev method is the same as the reference price $V_{0}$. Similarly to the equity Bermudan option, the expected exposure decreases over time and the PFE increases first and then decreases. The difference in the profiles comes from the number of exercise dates. Here, the swaption is only exerciseable once per year and the exposure jumps down at these days. In Table \ref{tab:EE_PFE_HW_Swaption} we see the corresponding relative errors for different Chebyshev N's. Here we used a dynamic Chebyshev method with higher accuracy to compute reference prices. As for an equity Bermudan option, an exact estimation of the exercise barrier is critical for a correct estimation of the exposure and leads to a higher $N$. However, since the volatility is lower, the interpolation domain is smaller and we need less nodes than for the Bermudan equity option. Figure \ref{fig:EE_PFE_Err_HW_Swaption_LSM} shows the relative error of the dynamic Chebyshev method and the LSM over the option's lifetime.

Table \ref{tab:EE_PFE_HW_Swaption_Runtime} shows the corresponding runtimes for $M=50000$ and $M=150000$ simulation paths of the underlying risk factor. Overall the runtime a slightly slower than for the equity products since the maturity is with five years much longer. The comparison of the dynamic Chebyshev method with the LSM reveals again a significantly higher efficiency. This is especially the case when it comes to the computation of the tail measure PFE.

\begin{figure}[H]
\begin{minipage}{.5\linewidth}
\centering
\subfloat[]{\includegraphics[scale=.48]{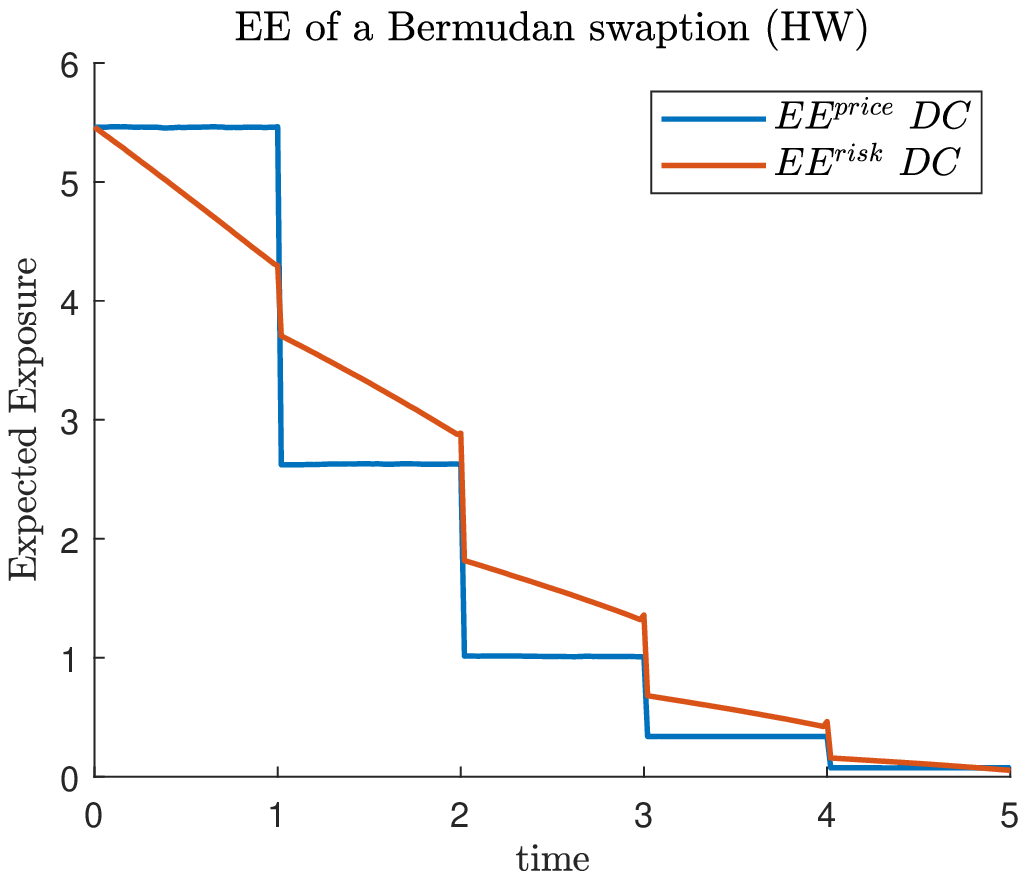}}
\end{minipage}%
\begin{minipage}{.5\linewidth}
\centering
\subfloat[]{\includegraphics[scale=.48]{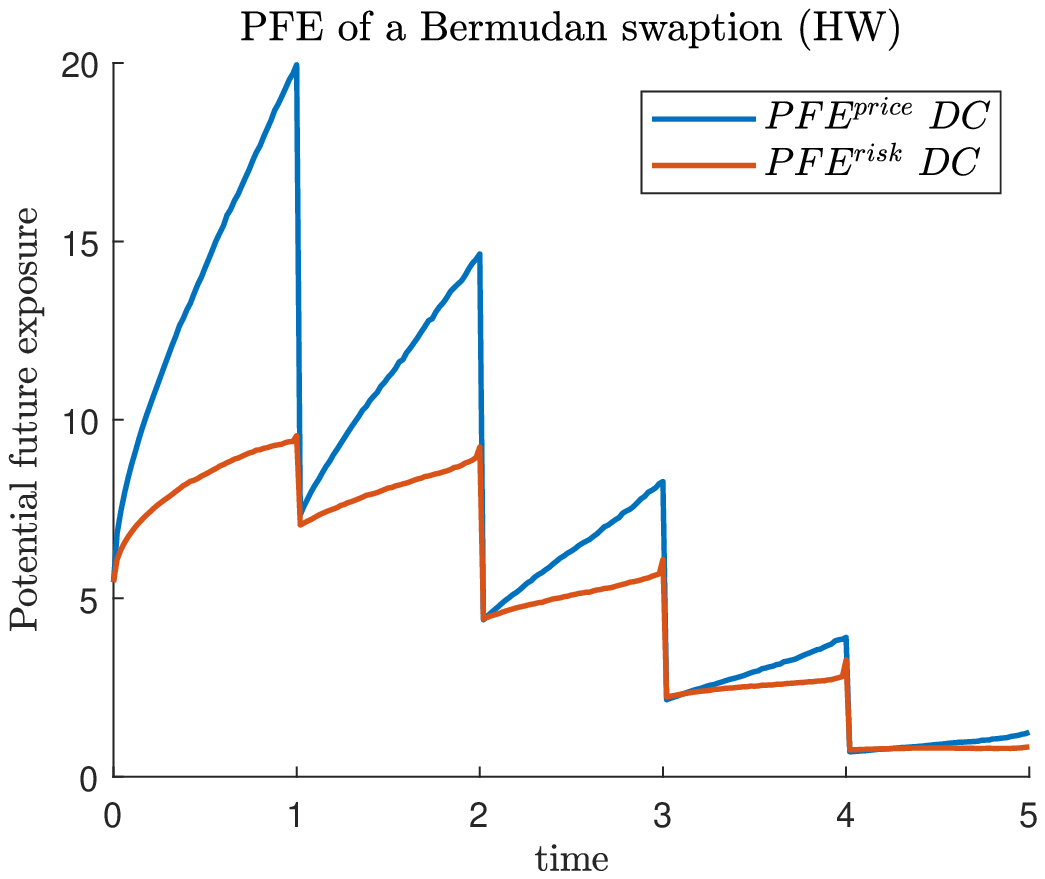}}
\end{minipage}
\caption{Expected exposure (left figure) and potential future exposure (right figure) of a Bermudan put option in the Merton jump-diffusion model, calculated under the pricing measure $\mathbb{Q}$ and the real world measure $\mathbb{P}$ using $M=150000$ simulations.}
\label{fig:EE_PFE_HW_Swaption}
\end{figure}

\begin{figure}[H]
\begin{minipage}{.5\linewidth}
\centering
\subfloat[]{\includegraphics[scale=.48]{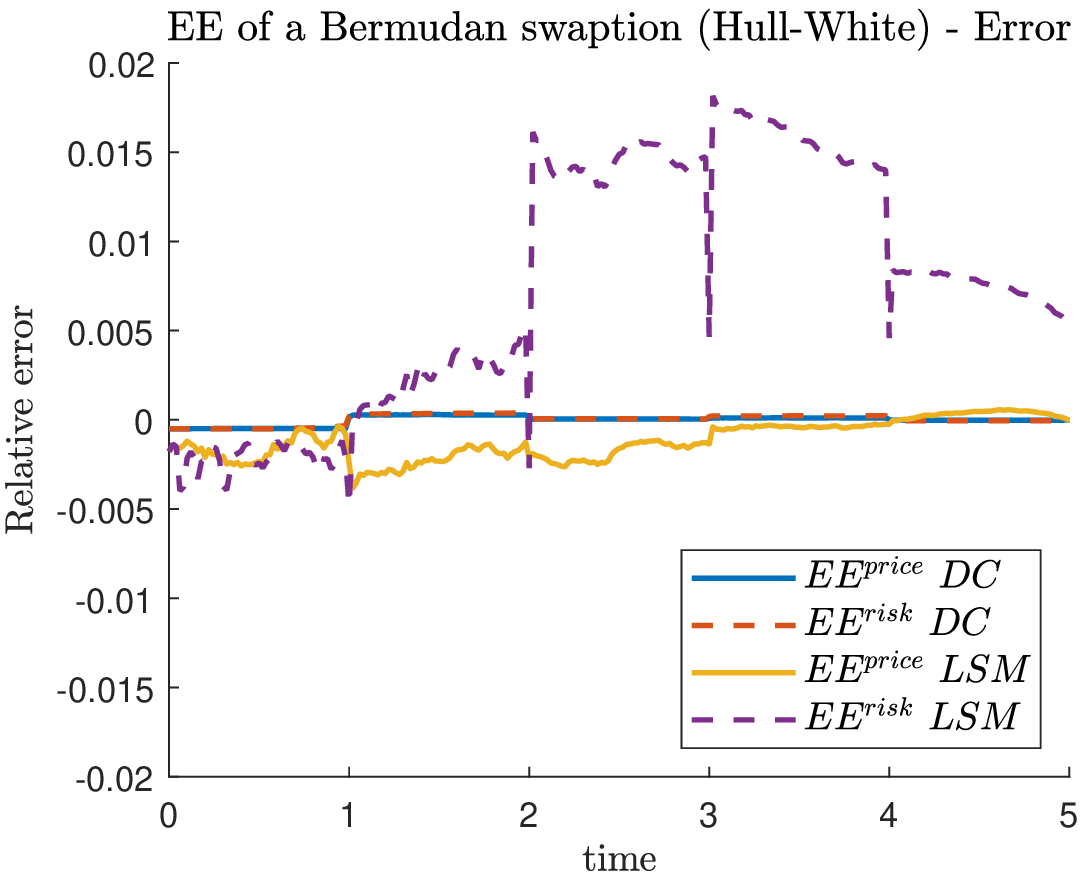}}
\end{minipage}%
\begin{minipage}{.5\linewidth}
\centering
\subfloat[]{\includegraphics[scale=.48]{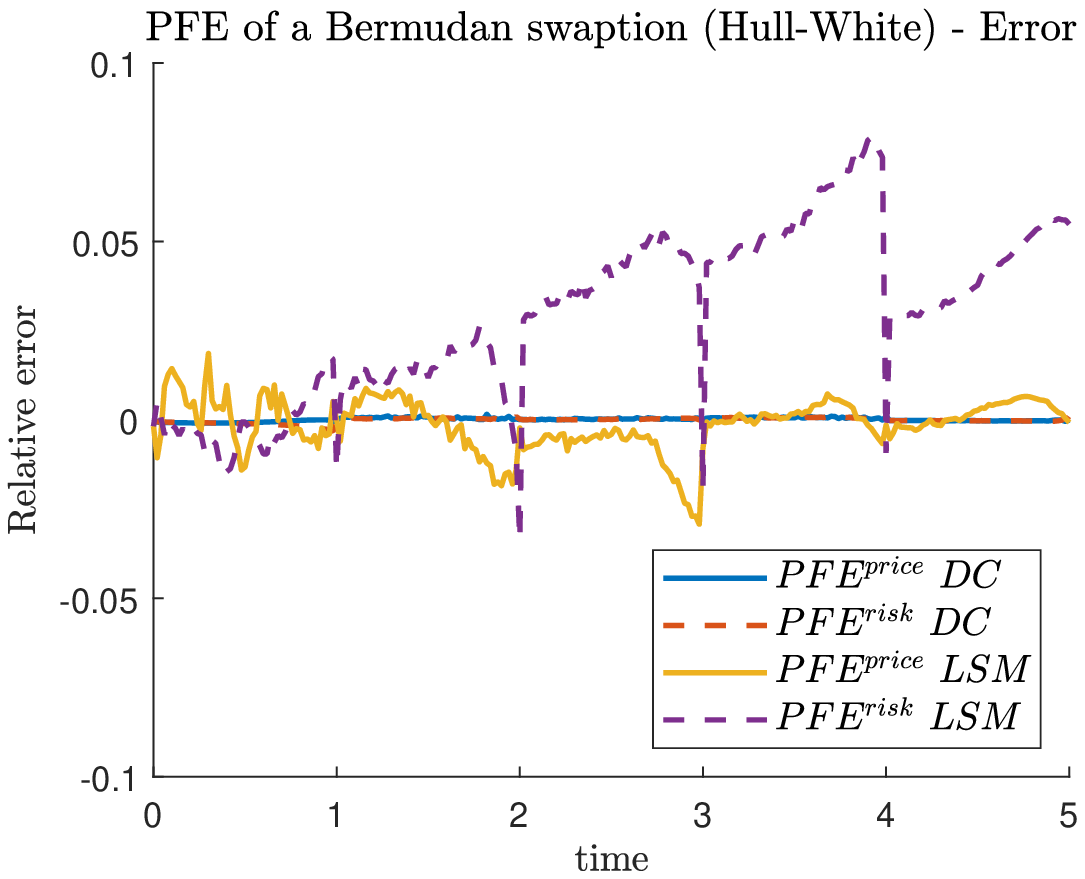}}
\end{minipage}
\caption{Relative error w.r.t. to the initial option price of the expected exposure (left figure) and potential future exposure (right figure) of a Bermudan receiver swaption in the Hull-White model of the dynamic Chebyshev method and the least-squares method. Calculated under the pricing measure $\mathbb{Q}$ and the real world measure $\mathbb{P}$ using $M=150000$ simulations.}
\label{fig:EE_PFE_Err_HW_Swaption_LSM}
\end{figure}

\begin{table}[H]
\begin{center}
\begin{tabular}{ccccc}
\hline 
Price & $EE^{price}_{T}$ & $PFE^{price}$ & $EE^{risk}_{T}$ & $PFE^{risk}$\\
5.4628 & 0.0771 & 1.2489 & 0.0540 & 0.8348\\
\hline 
\end{tabular} 
\caption{Reference values for option price, EE and PFE of a Bermudan receiver swaption in the Hull-White model using $M=150000$ simulations.}
\label{tab:EE_PFE_HW_Swaption_Reference} 
\end{center}
\end{table}

\begin{table}[H]
\begin{center}
\begin{tabular}{lccccc}
\hline 
 & Price & $EE^{price}$ & $PFE^{price}_{T}$ & $EE^{risk}$ & $PFE^{risk}_{T}$\\
\hline
$DC_{32}$ & 0.0148 & 0.0165 & 0.0444 & 0.0245 & 0.0351\\
$DC_{64}$ & 0.0026 & 0.0026 & 0.0069 & 0.0031 & 0.0166\\
$DC_{128}$ & 0.0005 & 0.0005 & 0.0016 & 0.0005 & 0.0032\\
$DC_{16, 16}$ & 0.0010 & 0.0022 & 0.0323 & 0.0034 & 0.0563\\
$DC_{32, 32}$ & 0.0005 & 0.0005 & 0.0015 & 0.0006 & 0.0026\\
$DC_{64, 64}$ & 0.0002 & 0.0002 & 0.0016 & 0.0002 & 0.0014\\
LSM & 0.0069 & 0.0277 & 0.1463 & 0.0331 & 0.1533\\
\hline 
\end{tabular} 
\caption{Maximal relative error w.r.t. to the initial stock price of option price, EE and PFE of a Bermudan put option in the Merton jump-diffusion model for $M=150000$ simulations. Comparison of the dynamic Chebyshev approach for different $N$ and the LSM approach with a full re-evaluation.}
\label{tab:EE_PFE_HW_Swaption} 
\end{center}
\end{table}

\begin{table}[H]
\begin{center}
\begin{tabular}{llcccccccc}
\hline 
& Sim. & $DC_{32}$ & $DC_{64}$ & $DC_{128}$ & $DC_{16,16}$ & $DC_{32,32}$ & $DC_{64,64}$ & LSM & $DC$ ref.\\  
\hline 
\multirow{2}{*}{$\mathbb{Q}$}
& $50k$ & 0.76s & 0.83s & 1.20s & 1.06s & 1.12s & 1.27s & 5.55s & 1.83s\\
& $150k$ & 2.58s & 2.83s & 3.56s & 3.06s & 3.22s & 3.55s & 8.56s & 4.77s\\
\multirow{2}{*}{$\mathbb{P}$}
& $50k$ & 0.66s & 0.76s & 1.16s & 0.97s & 1.02s & 1.17s & 5.47s & 1.88s\\
& $150k$ & 2.30s & 2.56s & 3.38s & 2.79s & 2.93s & 3.24s & 8.46s & 4.76s\\
\hline
\end{tabular} 
\caption{Runtimes of the exposure calculation of a barrier call option using the dynamic Chebyshev method for different $N$.}
\label{tab:EE_PFE_HW_Swaption_Runtime} 
\end{center}
\end{table}

\subsection{Summary of the experiments}
In this section, we analysed the dynamic Chebyshev method for credit exposure calculation numerically.

\cite{GlauMahlstedtPoetz2019} have validated the method for the pricing of options in different asset models. The experiments of this section show that the method is moreover well suited for credit exposure calculation of path-dependent options such as Bermudan and barrier equity options and Bermudan swaptions. Our examples show that the method can be applied to different models which require different numerical techniques for the calculation of the conditional expectations of the Chebyshev polynomials. 

The experiments show that the dynamic Chebyshev method is able to produce stable and accurate results even for the tail measure, the potential future exposure. It can handle the measure change from the pricing measure to the real-world measure without an additional computational effort and is therefore suited for the credit exposure calculation in both, pricing and risk management. 

The comparison with the popular LSM approach revealed the efficiency of the method in terms of accuracy versus runtime. This is especially the case for the computation of the potential future exposure. The LSM was not able to produce accurate prices in the tail for early-exercise options, which has also been observed in \cite{FengJainKarlssonKandhaiOosterlee2016}. Pricing methods based on Monte Carlo simulation and regression add additional simulation noise to the exposure calculation which is omitted in the new approach. Another methods that can improve the accuracy of the exposure in comparison to the LSM is the stochastic grid bundling method, as applied in \cite{KarlssonJainOosterlee2016} and \cite{FengJainKarlssonKandhaiOosterlee2016}. The experiments in \cite{FengJainKarlssonKandhaiOosterlee2016} however show, that this method is more than a factor of two times slower for the exposure calculation of a Bermudan swaption in the Hull-White model compared to a least-squares Monte Carlo approach. This allows us to conclude that the method will also be slower than the dynamic Chebyshev method.

Moreover, the experiments show that introducing an additional splitting in the dynamic Chebyshev method reduces the number of nodal points and can improve the efficiency of the exposure calculation further.  This is mainly interesting for large interpolation domains and early-exercise options.

\section{Conclusion and Outlook}
In this paper we have introduced a new method for the pricing and exposure calculation of European, Bermudan and barrier options based on the dynamic Chebyshev method of \cite{GlauMahlstedtPoetz2019}. The numerical experiments in Section 4 show that the method is well-suited for the exposure calculation and reveal several advantages.
\begin{itemize}
\item \textit{Flexibility}: The method offers a high flexibility, it applies to pricing and credit exposure calculation. The structure of the method allows us to explore additional knowledge of the model by choosing different techniques to compute the conditional expectations in a the pre-computation step. If the underlying  is conditionally normally distributed there is a closed form expression for these conditional expectations.
\item \text{Accuracy}: The method produces accurate exposure profiles and the approximation error is stable over the option's lifetime. In contrast to a least-squares Monte Carlo approach, the accuracy holds also in the tail of the distribution and for risk-measures such as the potential future exposure. The accuracy achieved is comparable to the one of a full re-evaluation.
\item \textit{Speed}: The comparison with the least-squares Monte Carlo approach showed the dynamic Chebyshev method is faster for all tested products and for both quantities, the expected exposure and the potential future exposure. As shown in \cite{FengJainKarlssonKandhaiOosterlee2016} the least-squares Monte Carlo is already competitive fast and outperforms the stochastic grid bundling method in terms of speed. For each simulated scenario, only a weighted sum of polynomials needs to be evaluated. Compared to a full re-evaluation repeated calls of a numerical pricing routine are avoided. This leads to a speed-up between one and two orders of magnitude in our experiments.
\end{itemize}
Overall, the combination of flexibility, accuracy and speed makes the dynamic Chebyshev method a highly efficient approach to compute credit exposure profiles for pricing and risk management.

Besides the confirmed quantitative advantages the method also admits several qualitative advantages. In Section 3 we provided an error analysis for the method which shows that the pricing error in the maximums norm decreases exponentially fast for analytic pricing functions. Moreover, Chebyshev interpolation exhibits algebraic convergence for differentiable functions and the derivatives converge as well. These estimates hold also true for the expected exposure and the potential future exposure.

The polynomial structure of the approximation of the value function allows not only a fast evaluation but enables also an efficient computation of the option's sensitivities Delta and Gamma in every time step. Additionally, the structure can be exploited to calculate the exposure of several options on the same underlying in one run of the price. Moreover, it allows the aggregation of credit exposures on different levels and facilitates the efficient computation of CVA and other risk metrics on a portfolio level. 

In this paper we focussed on the exposure calculation for products which depend only on one main risk factor. As a next step one can extend the presented approach for the exposure calculation to options which have more than one main risk factor. This can be tackled by tensorized Chebyshev interpolation for two or three risk factors. For higher dimensional problems, dimension reduction techniques such as sparse grids and low-rank tensor techniques can be exploited in the combination with Chebyshev interpolation. Moreover instead of multivariate Chebyshev interpolation other function approximation techniques can be used, for instance with kernel techniques We expect that these approaches achieve an accuracy comparable to a full re-evaluation in a significantly lower runtime.

The resulting reliable exposure profiles can further be used to learn counterparty credit risk measures on a portfolio level. The benefit of the presented method and its multivariate extensions would be the avoidance of nested Monte Carlo simulation in the trainings phase.

\appendix
\section{Proof of Proposition 4.2}
\begin{proof}
We define $\mu_{j}:=\EE[T_{j}(Y)\1_{[-1,1]}(Y)]$ as the generalized moments and $\mu^{\prime}_{j}=\EE[T^{\prime}_{j}(Y)\1_{[-1,1]}]$ as the expectations of the derivatives of the Chebyshev polynomials. The first three Chebyshev polynomials are given by $T_{0}(x)=1$, $T_{1}(x)=x$ and $T_{2}(x)=2x^{2}-1$ with derivatives $T_{0}^{\prime}(x)=0$, $T_{1}^{\prime}(x)=1=T_{0}(x)$ and $T_{2}^{\prime}(x)=4x=4T_{1}(x)$. This yields
\begin{align*}
\mu_{0}=\EE[1_{[-1,1]}(Y)]&=P(-1\leq Y\leq 1)=F(1)-F(-1).
\end{align*}
Before we consider the first moment we need the following property of the density $f$ of the normal distribution,
\begin{align*}
f^{\prime}(x)&=\frac{1}{\sqrt{2\pi}\sigma}e^{-\frac{(x-mu)^{2}}{2\sigma^{2}}}(-2\frac{(x-\mu)}{2\sigma^{2}})=f(x)(-2\frac{(x-\mu)}{2\sigma^{2}})=(-\frac{1}{\sigma^{2}})xf(x)+\frac{\mu}{\sigma^{2}}f(x),\\
&\text{and hence}\quad xf(x)=\mu f(x) - \sigma^{2}f^{\prime}(x).
\end{align*}
Using this property we obtain for the first moment $\mu_{1}=\EE[Y1_{[-1,1]}(Y)]$
\begin{align*}
\mu_{1}=\int_{-1}^{1}yf(y)\dy=\mu\int_{-1}^{1}f(y)\dy-\sigma^{2}\int_{-1}^{1}f^{\prime}(y)\dy=\mu\mu_{0}-\sigma^{2}(f(1)-f(-1)).
\end{align*}
Assume we know $\mu_{j},\mu^{\prime}_{j}$, $j=0,\ldots,n$. The Chebyshev polynomials and their derivative are recursively given by
\begin{align*}
T_{n+1}(x)=2xT_{n}(x)-T_{n-1}(x)\qquad T^{\prime}_{n+1}(x)=2(n+1)T_{n}(x)+\frac{n+1}{n-1}T^{\prime}_{n-1}(x).
\end{align*}
From the latter easily follows that
\begin{align*}
\mu^{\prime}_{n+1}&=\EE[T^{\prime}_{n+1}(Y)\1_{[-1,1]}(Y)]\\
&=2(n+1)\EE[T_{n}(Y)\1_{[-1,1]}(Y)]+\frac{n+1}{n-1}\EE[T^{\prime}_{n-1}(Y)\1_{[-1,1]}(Y)]\\
&=2(n+1)\mu_{n}+\frac{(n+1)}{(n-1)}\mu^{\prime}_{n-1}
\end{align*}
for $n\geq 2$. For the generalized moments we obtain
\begin{align*}
\mu_{n+1}=\EE[T_{n+1}(Y)\1_{[-1,1]}(Y)]=2\EE[Y T_{n}(Y)\1_{[-1,1]}(Y)]-\EE[T_{n-1}(Y)\1_{[-1,1]}(Y)].
\end{align*}
The second term is simply $\mu_{n-1}$ and for the first term we obtain
\begin{align*}
\EE[Y T_{n}\1_{[-1,1]}(Y)]&=\int_{-1}^{1}yT_{n}(y)f(y)\dy\\
&=\mu\int_{-1}^{1}T_{n}(y)f(y)\dy - \sigma^{2}\int_{-1}^{1}T_{n}(y)f^{\prime}(y)\dy\\
&=\mu\mu_{n} - \sigma^{2}\Big(T_{n}(y)f(y)\Big\vert_{-1}^{1}-\int_{-1}^{1}T^{\prime}_{n}(y)f(y)\dy\Big)\\
&=\mu\mu_{n}-\sigma^{2}\big(T_{n}(1)f(1)-T_{n}(-1)f(-1)-\mu^{\prime}_{n}\big).
\end{align*}
Altogether we obtain
\begin{align*}
\mu_{n+1}&=2\EE[Y T_{n}(Y)\1_{[-1,1]}(Y)]-\EE[T_{n-1}(Y)\1_{[-1,1]}(Y)]\\
&=2\big(\mu\mu_{n}-\sigma^{2}\big(T_{n}(1)f(1)-T_{n}(-1)f(-1)-\mu^{\prime}_{n}\big)\big)-\mu_{n-1}.
\end{align*}
It remains to find an expression for $\mu^{\prime}_{n}$.\\

We will prove by induction that
\begin{align}\label{eq:Cheby_deriv_recursive}
\mu^{\prime}_{n+1}=2(n+1)\sum_{j=0}^{n}{}^{'}\mu_{j}\1_{(n+j)\bmod 2=0}, \quad n\geq 0
\end{align}
where $\sum{}^{'}$ indicates that the first term is multiplied with $1/2$. For $n=0$, we obtain
\begin{align*}
\mu^{\prime}_{1}=2\sum_{j=0}^{0}{}^{'}\mu_{j}\1_{(0+j)\bmod 2=0}=2\frac{1}{2}\mu_{0}=\1_{0\bmod 2=0}=\mu_{0}.
\end{align*}
which shows \eqref{eq:Cheby_deriv_recursive}. Assume \eqref{eq:Cheby_deriv_recursive} holds for $j=0,\ldots,n$. Then we obtain
\begin{align*}
\mu^{\prime}_{n+1}&=2(n+1)\mu_{n}+\frac{(n+1)}{(n-1)}\mu^{\prime}_{n-1}\\
&=2(n+1)\mu_{n}+\frac{(n+1)}{(n-1)}2(n-1)\sum_{j=0}^{n-2}{}^{'}\mu_{j}\1_{(n-2+j)\bmod 2=0}\\
&=2(n+1)\Big(\mu_{n}\1_{(n+n)\bmod 2=0}+\mu_{n-1}\1_{(n+n-1)\bmod 2=0}+\sum_{j=0}^{n-2}{}^{'}\mu_{j}\1_{(n+j)\bmod 2=0}\Big)\\
&=2(n+1)\sum_{j=0}^{n}{}^{'}\mu_{j}\1_{(n+j)\bmod 2=0}.
\end{align*}
We use that $(n+j)\bmod 2=(2+j-2)\bmod 2$. For the generalized moments we thus obtain
\begin{align*}
\mu_{n+1}=2\mu\mu_{n} - 2\sigma^{2}\big(f(1)-f(-1)T_{n}(-1)-2n\sum_{j=0}^{n-1}{}^{'}\mu_{j}\1_{(n+j)\bmod 2=1}\big)-\mu_{n-1}
\end{align*}
which was our claim.
\end{proof}
\bibliographystyle{chicago}
  \bibliography{CVA_Literature}

\begin{thebibliography}{}

\bibitem[\protect\citeauthoryear{Black and Scholes}{Black and
  Scholes}{1973}]{BlackScholes1973}
Black, F. and M.~Scholes (1973).
\newblock The pricing of options and other liabilities.
\newblock {\em Journal of Political Economy\/}~{\em 81}, 637--654.

\bibitem[\protect\citeauthoryear{Brigo and Mercurio}{Brigo and
  Mercurio}{2007}]{BrigoMercurio2007}
Brigo, D. and F.~Mercurio (2007).
\newblock {\em Interest rate models-theory and practice: with smile, inflation
  and credit}.
\newblock Springer Science \& Business Media.

\bibitem[\protect\citeauthoryear{Fang and Oosterlee}{Fang and
  Oosterlee}{2009}]{FangOosterlee2009}
Fang, F. and C.~W. Oosterlee (2009).
\newblock Pricing early-exercise and discrete barrier options by fourier-cosine
  series expansions.
\newblock {\em Numerische Mathematik\/}~{\em 114\/}(1), 27.

\bibitem[\protect\citeauthoryear{Feng, Jain, Karlsson, Kandhai, and
  Oosterlee}{Feng et~al.}{2016}]{FengJainKarlssonKandhaiOosterlee2016}
Feng, Q., S.~Jain, P.~Karlsson, D.~Kandhai, and C.~W. Oosterlee (2016).
\newblock Efficient computation of exposure profiles on real-world and
  risk-neutral scenarios for bermudan swaptions.
\newblock {\em Journal of Computational Finance\/}~{\em 20\/}(1), 139--172.

\bibitem[\protect\citeauthoryear{Ga{\ss}, Glau, Mahlstedt, and Mair}{Ga{\ss}
  et~al.}{2018}]{GassGlauMahlstedtMair2018}
Ga{\ss}, M., K.~Glau, M.~Mahlstedt, and M.~Mair (2018).
\newblock Chebyshev interpolation for parametric option pricing.
\newblock {\em Finance and Stochastics\/}~{\em 22\/}(3), 701--731.

\bibitem[\protect\citeauthoryear{Glau, Mahlstedt, and P{\"o}tz}{Glau
  et~al.}{2019}]{GlauMahlstedtPoetz2019}
Glau, K., M.~Mahlstedt, and C.~P{\"o}tz (2019).
\newblock A new approach for {A}merican option pricing: The {D}ynamic
  {C}hebyshev method.
\newblock {\em SIAM Journal on Scientific Computing\/}~{\em 41\/}(1),
  B153--B180.

\bibitem[\protect\citeauthoryear{Green}{Green}{2015}]{Green2015}
Green, A. (2015).
\newblock {\em XVA: Credit, Funding and Capital Valuation Adjustments}.
\newblock John Wiley \& Sons.

\bibitem[\protect\citeauthoryear{Gregory}{Gregory}{2010}]{Gregory2010}
Gregory, J. (2010).
\newblock {\em Counterparty credit risk: The new challenge for global financial
  markets}, Volume 470.
\newblock John Wiley \& Sons.

\bibitem[\protect\citeauthoryear{Jain and Oosterlee}{Jain and
  Oosterlee}{2015}]{JainOosterlee2015}
Jain, S. and C.~W. Oosterlee (2015).
\newblock The stochastic grid bundling method: Efficient pricing of bermudan
  options and their greeks.
\newblock {\em Applied Mathematics and Computation\/}~{\em 269}, 412--431.

\bibitem[\protect\citeauthoryear{Karlsson, Jain, and Oosterlee}{Karlsson
  et~al.}{2016}]{KarlssonJainOosterlee2016}
Karlsson, P., S.~Jain, and C.~W. Oosterlee (2016).
\newblock Counterparty credit exposures for interest rate derivatives using the
  stochastic grid bundling method.
\newblock {\em Applied Mathematical Finance\/}~{\em 23\/}(3), 175--196.

\bibitem[\protect\citeauthoryear{Longstaff and Schwartz}{Longstaff and
  Schwartz}{2001}]{LongstaffSchwartz2001}
Longstaff, F.~A. and E.~S. Schwartz (2001).
\newblock Valuing {A}merican options by simulation: {A} simple least-squares
  approach.
\newblock {\em The review of financial studies\/}~{\em 14\/}(1), 113--147.

\bibitem[\protect\citeauthoryear{Merton}{Merton}{1976}]{Merton1976}
Merton, R.~C. (1976).
\newblock Option pricing when underlying stock returns are discontinuous.
\newblock {\em Journal of financial economics\/}~{\em 3\/}(1-2), 125--144.

\bibitem[\protect\citeauthoryear{Sauter and Schwab}{Sauter and
  Schwab}{2010}]{SauterSchwab2010}
Sauter, S. and C.~Schwab (2010).
\newblock {\em Boundary Element Methods, Translated and expanded from the 2004
  German original}, Volume~39.
\newblock Springer Series Computational Mathematics.

\bibitem[\protect\citeauthoryear{Sch{\"o}ftner}{Sch{\"o}ftner}{2008}]{Schoeftner2008}
Sch{\"o}ftner, R. (2008).
\newblock On the estimation of credit exposures using regression-based {M}onte
  {C}arlo simulation.
\newblock {\em The journal of credit risk\/}~{\em 4\/}(4), 37--62.

\bibitem[\protect\citeauthoryear{Shen, Van Der~Weide, and Anderluh}{Shen
  et~al.}{2013}]{ShenWeideAnderluh2013}
Shen, Y., J.~A. Van Der~Weide, and J.~H. Anderluh (2013).
\newblock A benchmark approach of counterparty credit exposure of bermudan
  option under {L}{\'e}vy process: the monte carlo-cos method.
\newblock {\em Procedia Computer Science\/}~{\em 18}, 1163--1171.

\bibitem[\protect\citeauthoryear{Stein}{Stein}{2016}]{Stein2016}
Stein, H.~J. (2016).
\newblock Fixing risk neutral risk measures.
\newblock {\em International Journal of Theoretical and Applied Finance\/}~{\em
  19\/}(03), 1650021.

\bibitem[\protect\citeauthoryear{Trefethen}{Trefethen}{2013}]{Trefethen2013}
Trefethen, L.~N. (2013).
\newblock {\em {Approximation Theory and Approximation Practice}}.
\newblock SIAM books.

\bibitem[\protect\citeauthoryear{von Sydow, Josef~H{\"o}{\"o}k, Larsson,
  Lindstr{\"o}m, Milovanovi{\'c}, Persson, Shcherbakov, Shpolyanskiy,
  Sir{\'e}n, Toivanen, et~al.}{von Sydow et~al.}{2015}]{Benchop2015}
von Sydow, L., L.~Josef~H{\"o}{\"o}k, E.~Larsson, E.~Lindstr{\"o}m,
  S.~Milovanovi{\'c}, J.~Persson, V.~Shcherbakov, Y.~Shpolyanskiy,
  S.~Sir{\'e}n, J.~Toivanen, et~al. (2015).
\newblock Benchop--the benchmarking project in option pricing.
\newblock {\em International Journal of Computer Mathematics\/}~{\em 92\/}(12),
  2361--2379.

\end{thebibliography}

\end{document}